\newif\ificml
\icmlfalse

\documentclass{article}

\usepackage{microtype}
\usepackage{graphicx}
\usepackage{subfigure}
\usepackage{booktabs} 

\usepackage{hyperref}

\PassOptionsToPackage{noend}{algorithmic} 

\usepackage[accepted]{icml2024}

\usepackage{amsmath}
\usepackage{amssymb}
\usepackage{mathtools}
\usepackage{amsthm}

\usepackage[capitalize,noabbrev,nameinlink]{cleveref}

\theoremstyle{plain}
\newtheorem{theorem}{Theorem}[section]

\newtheorem{lemma}[theorem]{Lemma}
\newtheorem{corollary}[theorem]{Corollary}
\theoremstyle{definition}
\newtheorem{definition}[theorem]{Definition}
\newtheorem{assumption}[theorem]{Assumption}
\theoremstyle{remark}

\newtheorem{observation}[theorem]{Observation}
\newtheorem{claim}[theorem]{Claim}

\usepackage{xspace}
\usepackage{enumitem}

\newcommand{\cX}{\mathcal{X}}
\newcommand{\cM}{\mathcal{M}}
\newcommand{\cA}{\mathcal{A}}
\newcommand{\cI}{\mathcal{I}}

\newcommand{\ind}{\mathbf{1}}
\newcommand{\eps}{\varepsilon}
\newcommand{\N}{\mathbb{N}}

\newcommand{\R}{\mathbb{R}}

\newcommand{\E}{\mathbb{E}}

\newcommand{\cP}{\mathcal{P}}

\newcommand{\cS}{\mathcal{S}}
\newcommand{\cU}{\mathcal{U}}

\newcommand{\cH}{\mathcal{Y}}

\newcommand{\by}{\mathbf{y}}
\newcommand{\bz}{\mathbf{z}}

\newcommand{\bx}{\mathbf{x}}

\renewcommand{\phi}{\varphi}

\newcommand{\cost}{\mathrm{cost}}
\newcommand{\tih}{\tilde{h}}
\newcommand{\tD}{\tilde{D}}
\newcommand{\hD}{\hat{D}}
\newcommand{\hath}{\hat{h}}

\newcommand{\tn}{\tilde{n}}

\newcommand{\tO}{\tilde{O}}

\newcommand{\cC}{\mathcal{C}}

\newcommand{\cO}{\mathcal{O}}

\DeclareMathOperator{\Lap}{Lap}
\DeclareMathOperator{\Exp}{Exp}

\DeclareMathOperator{\argmin}{argmin}

\newcommand{\subs}{\mathrm{s}}
\newcommand{\EM}{\textsc{ExpMech}}
\newcommand{\smcalg}{\textsc{DPSubmodGreedy}}
\newcommand{\setcovalg}{\textsc{DPGreedyScaling}}
\newcommand{\rem}{\textsc{Repeated-EM}}
\newcommand{\rat}{\textsc{Repeated-AT}}
\newcommand{\optsc}{\mathrm{SetCov}}
\newcommand{\setcov}{\mathrm{SetCov}}
\newcommand{\costsc}{\mathrm{CostSetCov}}
\newcommand{\threshmon}{\textsc{ThreshMonitor}}
\newcommand{\mult}{\mathrm{mult}}
\newcommand{\swapr}{\mathrm{SwapRound}}
\newcommand{\contgreedy}{\mathrm{PrivContGreedy}}
\newcommand{\rank}{\mathrm{rank}}
\newcommand{\opt}{\mathrm{OPT}}
\newcommand{\RETURN}{\STATE \textbf{return}~}
\renewcommand{\setminus}{\smallsetminus}

\renewcommand{\algorithmicrequire}{\textbf{Input:}} 
\renewcommand{\algorithmicensure}{\textbf{Output:}} 
\newcommand{\PARAMETERS}{\STATE \hspace{-3.5mm}{\bf Parameters:}\xspace}

\allowdisplaybreaks

\definecolor{Gred}{RGB}{219, 50, 54}
\definecolor{Ggreen}{RGB}{60, 186, 84}
\definecolor{Gblue}{RGB}{72, 133, 237}
\definecolor{Gyellow}{RGB}{247, 178, 16}
\definecolor{ToCgreen}{RGB}{0, 128, 0}
\definecolor{myGold}{RGB}{231,141,20}
\definecolor{myBlue}{rgb}{0.19,0.41,.65}
\definecolor{myPurple}{RGB}{175,0,124}

\providecommand{\Comments}{0}
\ifnum\Comments>0
\usepackage[colorinlistoftodos,prependcaption,textsize=scriptsize]{todonotes}
\paperwidth=\dimexpr \paperwidth + 4.7cm\relax
\oddsidemargin=\dimexpr\oddsidemargin + 2.6cm\relax
\evensidemargin=\dimexpr\evensidemargin + 2.6cm\relax
\marginparwidth=\dimexpr\marginparwidth + 1.7cm\relax
\else
\usepackage[disable]{todonotes}
\fi
\newcommand{\mytodo}[1]{\ifnum\Comments=1{#1}\fi}

\newcommand{\pritishimp}[1]{\ifnum\Comments<4\todo[linecolor=Gred,backgroundcolor=Gred!25,bordercolor=Gred]{Pritish: #1}\fi}
\newcommand{\pritish}[1]{\ifnum\Comments<3\todo[linecolor=myGold,backgroundcolor=myGold!25,bordercolor=myGold]{Pritish: #1}\fi}
\newcommand{\pritishinfo}[1]{\ifnum\Comments<2\todo[linecolor=Ggreen,backgroundcolor=Ggreen!25,bordercolor=Ggreen]{Pritish: #1}\fi}

\newcommand{\tableoftodos}{\ifnum\Comments=1 \listoftodos[Comments/To Do's] \fi}

\icmltitlerunning{Individualized Privacy Accounting via Subsampling}

\begin{document}

\twocolumn[
\icmltitle{Individualized Privacy Accounting via Subsampling \\ with Applications in Combinatorial Optimization}




\begin{icmlauthorlist}
\icmlauthor{Badih Ghazi}{gmtv}
\icmlauthor{Pritish Kamath}{gmtv}
\icmlauthor{Ravi Kumar}{gmtv}
\icmlauthor{Pasin Manurangsi}{gth}
\icmlauthor{Adam Sealfon}{gny}
\end{icmlauthorlist}

\icmlaffiliation{gmtv}{Google Research, Mountain View, USA}
\icmlaffiliation{gny}{Google Research, New York, USA}
\icmlaffiliation{gth}{Google Research, Thailand}

\icmlcorrespondingauthor{Pasin Manurangsi}{pasin@google.com}


\vskip 0.3in
]



\printAffiliationsAndNotice{\icmlEqualContribution} 

\begin{abstract}
In this work, we give a new technique for analyzing individualized privacy accounting via the following simple observation: if an algorithm is one-sided add-DP, then its subsampled variant satisfies two-sided DP. From this, we obtain several improved algorithms for private combinatorial optimization problems, including decomposable submodular maximization and set cover. Our error guarantees are asymptotically tight and our algorithm satisfies \emph{pure-}DP while previously known algorithms~\cite{GLMRT10,CNZ21} are \emph{approximate-}DP. We also show an application of our technique beyond combinatorial optimization by giving a pure-DP algorithm for the shifting heavy hitter problem in a stream; previously, only an approximate-DP algorithm was known~\cite{KMS21,CL23}.
\end{abstract}

\section{Introduction}

In combinatorial optimization, we typically wish to select a discrete object to minimize or maximize  certain objective functions subject to certain constraints. 
%
%
In several settings, such objective functions or constraints may depend on sensitive information of users.
For example, clustering and facility location tasks may involve taking users' location information as part of the objectives or constraints. Similarly, data summarization may require user-produced examples as part of the objective~\cite{MirzasoleimanBK16}. Due to this, several works have considered studying combinatorial optimization problems under \emph{differential privacy (DP)}~\cite{dwork2006calibrating,dwork2006our}---a widely-used and rigorous notion to quantify privacy properties of an algorithm. To state the definition, we use $\cX$ to denote the domain of each user's data. Two datasets $D, D' \subseteq \cX^*$ are said to be \emph{add-remove neighbors} if $D$ is a result of adding an element to $D'$ or removing an element from $D'$.

\begin{definition}[Add-remove DP,~\citet{dwork2006calibrating}] \label{def:dp}
A randomized algorithm $\cM: \cX^* \to \cO$ is \emph{$(\eps, \delta)$-DP} if, for every add-remove neighboring datasets $D, D'$ and every set $S \subseteq \cO$ of outcomes, we have \\
\centerline{$\Pr[\cM(D) \in S] \leq e^{\eps} \cdot \Pr[\cM(D') \in S] + \delta.$}
\end{definition}

When $\delta > 0$, we say that the algorithm satisfies \emph{approximate-DP}; when $\delta = 0$, we say that the algorithm satisfies \emph{pure-DP}. The latter is preferable as it provides stronger privacy protection; more specifically, it does not allow for ``catastrophic failure'' where the sensitive input data is leaked in the clear.

\textbf{Submodular Maximization.}
In submodular maximization, we are given a submodular\footnote{A set function $F$ is \emph{submodular} if, for every $S \subseteq T \subseteq [m]$ and $v \in [m]$, we have $F(S \cup \{v\}) - F(S) \geq F(T\cup\{v\})-F(T)$.} set function $F: 2^{[m]} \to \R$, where $[m] = \{1, \dots, m\}$ is the universe. The goal is to find a subset $T \subseteq [m]$ that maximizes $F(T)$ under certain constraints. In this work, we consider two types of constraints:
\begin{itemize}[nosep]
\item \emph{Cardinality Constraint}: Here we want $|T| = k$.
\item \emph{Matroid Constraint}: Given a rank-$k$ matroid $M = ([m], \cI)$, we want $T$ to be an independent set of the matroid (i.e., $T \in \cI$). Note that this generalizes the cardinality constraint (when $M$ is the uniform matroid).   
\end{itemize}
Submodular maximization is among the most well-studied problems in combinatorial optimization; several algorithms date
back to the 70's~\cite{NemhauserWF78} and many variants continue to be studied to this day (e.g.,~\citet{DuettingFLNZ23,BanihashemBGHJM23}).

In this work, we consider the $1$-decomposable\footnote{More generally, \emph{$\lambda$-decomposable} refers to the same definition but with $f_x: 2^{[m]} \to [0, \lambda]$. All results discussed in this paper applied to $\lambda$-decomposable functions as well by appropriately scaling the functions.} monotone submodular maximization problem\footnote{Aka the \emph{Combinatorial Public Project} problem~\cite{PSS08}.}. Here each $x \in \cX$ is associated with a monotone submodular function $f_x: 2^{[m]} \to [0, 1]$ and the goal is to maximize $F_D := \sum_{x \in D} f_x$. 

For DP submodular maximization under cardinality constraint, \citet{GLMRT10} gave a polynomial-time $(\eps, \delta)$-DP algorithm that achieves\footnote{An output $T^*$ is said to achieve \emph{$\alpha$-approximation}  and \emph{$\kappa$-error} if $F(T^*) \geq \alpha \cdot \opt - \kappa$, where $\opt$ is the optimum.} $\left(1 - \frac{1}{e}\right)$-approximation and $O\left(\frac{k \log m \log(1/\delta)}{\eps}\right)$-error. They also show a lower bound of $\Omega\left(\frac{k \log m}{\eps}\right)$ on the error. Since $\left(1 - \frac{1}{e} + o(1)\right)$-approximation is NP-hard \cite{Feige98}, Gupta et al.'s result is tight up to the $O(\log(1/\delta))$ factor  in the error.
The matroid constraint case was first studied by \citet{MBKK17}, who gave an efficient $(\eps,\delta)$-DP $\frac{1}{2}$-approximation and the same error bound. Recently, \citet{CNZ21} improved the approximation ratio to $\left(1 - \frac{1}{e} - \eta\right)$ for any constant $\eta > 0$ while retaining the same error bound. Again, this is tight up to a factor of $O\left(\log(1/\delta)\right)$ in the error.


\textbf{Set Cover.}
In the \emph{Set Cover} problem, we are given a set system $(\cU, \cS = (S_1, \dots, S_m))$. The goal is to output as few sets as possible that cover the universe, i.e.,   $S_{i_1}, \dots, S_{i_k}$ such that $S_{i_1} \cup \cdots \cup S_{i_k} = \cU$. We use $\optsc(\cU, \cS)$ to denote the optimal size of the set cover. Set Cover can be viewed as a ``dual'' version of submodular maximization under cardinality constraint, since the coverage function\footnote{Maximizing the coverage function among $k$-size $I$ is known as the \emph{max $k$-Coverage} problem, which is a special case of submodular maximization with cardinality constraint and is also well studied in the literature.} $F(I) = |\bigcup_{i \in I} S_i|$ is submodular. Set Cover is a classic combinatorial optimization problem, being one of the 21 original NP-complete problems~\cite{Karp72}.

For private Set Cover, DP is w.r.t. adding or removing an item from the universe (and all the sets)\footnote{More precisely, we can define $\cX = \{0, 1\}^m$, where $x \in \cX$ belongs to all $S_i$ such that $x_i = 1$.}. 
Unfortunately, \citet{GLMRT10} show that no non-trivial approximation is possible for the setting where we output the indices $i_1 , \dots, i_k$ directly. Instead, they proposed what is now sometimes referred to as the \emph{open set} setting, where we instead output a permutation $\pi: [m] \to [m]$. Each element $x \in \cU$ then chooses the first set containing it in the sequence (i.e., $\min\{i \mid x \in S_{\pi(i)}\}$). 
The cost is then defined as the number of sets that are chosen; we use $\costsc_{\cU, \cS}(\pi)$ to denote the cost of $\pi$. We work in this model throughout the paper. 
Under this model, they provide an $(\eps, \delta)$-DP algorithm with an expected approximation ratio $O\left(\ln n + \frac{\ln m \log(1/\delta)}{\eps}\right)$ for the problem where $n$ denotes the size of the input dataset $|\cU|$. This is nearly tight as it is NP-hard to achieve an $o(\ln n)$-approximation~\cite{DinurS14,Moshkovitz15}, and \citet{GLMRT10} show that no $\eps$-DP algorithm can achieve an $o\left(\frac{\ln m}{\eps}\right)$-approximation.


\textbf{\boldmath Metric $k$-Means and $k$-Median.}
In the \emph{metric $(k, q)$-clustering} problem, there is a metric space $([m], d)$ whose diameter\footnote{The diameter is defined as $\max_{a, b \in [m]} d(a, b)$.} is at most one, each user input is a point in this metric space (i.e., $\cX = [m]$),  and the goal is to output a subset $S \subseteq [m]$ of size $k$ that minimizes $\cost^q(S; D) := \sum_{x \in D} \min_{c \in S} d(c, x)^q$.
When $q = 1$ and $q = 2$, this problem is referred to as \emph{metric $k$-median} and \emph{metric $k$-means} respectively. In the non-private setting, even though constant-factor approximation algorithms for these problems have long been known~\cite{CGTS99}, tight (hardness of) approximation ratios are not yet known and this remains a challenging and active area of research (see, e.g., \citet{anand2024separating} and the references therein).

For private $k$-median, \citet{GLMRT10} gave an $\eps$-DP algorithm with approximation ratio\footnote{They only claim an approx. ratio of 6 but it is straightforward to see that this can be extended to any approx. ratio greater than 5.} $(5 + \eta)$ for any constant $\eta > 0$ with error\footnote{We note that in both \cite{GLMRT10} and \cite{JNN21}, it was assumed that $n \leq m$. } $O\left(\frac{k^2 \log^2(mn)}{\eps}\right)$. On the lower bound front, they showed that any $\eps$-DP algorithm must incur error at least $\Omega(k \log n)$. Later, the error bound was improved by \citet{JNN21} to $O\left(\frac{k \log(mn) \log(1/\delta)}{\eps}\right)$, albeit with an $(\eps,\delta)$-DP algorithm. The approximation ratio of \citet{JNN21} is a constant, but for simplicity is not explicit.

\textbf{Shifting Heavy Hitters.} Our framework  will also apply to the ostensibly unrelated \emph{shifting heavy hitters} problem~\cite{KMS21}.
Here, each user $i$'s data $\bx_i$ is a stream $(x_{i,1}, \dots, x_{i, T}) \in \cH^T$ where $x_{i, t}$ is the ``bucket'' that the user contributes to at time $t$. For $\tau \geq 0, t \in [T]$, a \emph{$\tau$-heavy hitter} at time $t$ is an element $y \in \cH$ that appears at least $\tau$ times, i.e., $w_t(y) \geq \tau$ where $w_t(y) := |\{i \in [n] \mid x_{i, t} = y\}|$. 
Following \citet{KMS21}, an algorithm is said to have error $\tau$ with probability (w.p.) $1 - \beta$ if the following holds w.p. $1 - \beta$ for all $t \in [T]$:
\begin{itemize}[nosep]
\item Every reported element $x$ satisfies $w_t(x) > 0$.
\item Every $\tau$-heavy hitter is reported.
\end{itemize}

Without any additional assumption, the best error one can achieve with $(\eps, \delta)$-DP is $\tO(\sqrt{T \log(1/\delta)})$. The main result of \citet{KMS21} is that, under the assumption that each user contributes to at most $k$ heavy hitters, the error can be reduced to $\tO(\sqrt{k} \cdot \log(1/\delta) \log T)$.


All state-of-the-art algorithms we have discussed so far are \emph{approximate-DP}. Meanwhile, known pure-DP algorithms have significantly worse error guarantees; in fact, for some problems such as private set cover and shifting heavy hitters, no non-trivial pure-DP algorithms are known. This leads us to the main question of this work: \emph{Are there pure-DP algorithms that achieve similar (or even better) bounds?}

\subsection{Our Results}

We answer the above question positively by giving pure-DP algorithms for all above problems via a unified framework. As explained below, our error bounds are all nearly tight. In fact, for the optimization problems, we even improve on the error bounds  from previous approximate-DP algorithms. 

We note that all of our algorithms run in polynomial time and we will not state this explicitly below for brevity.

\textbf{Monotone Submodular Maximization.}
For monotone submodular maximization with a cardinality constraint, we can get an approximation ratio arbitrarily close to $1 - \frac{1}{e}$ while having an error $O\left(\frac{k \log m}{\eps}\right)$, as stated more precisely below. The former matches \cite{GLMRT10} where the latter improves on their bound by a factor of $O(\log(1/\delta))$ and is tight due to their $\Omega\left(\frac{k \log m}{\eps}\right)$ lower bound.

\begin{theorem} \label{thm:smc}
For any $0 < \eps, \beta, \eta < 1$, there is an $\eps$-DP algorithm for monotone submodular maximization under a cardinality constraint that achieves $\left(1 - \frac{1}{e} - \eta\right)$-approximation and $O\left(\frac{k \log (m/\beta)}{\eta \eps}\right)$-error w.p. $1 - \beta$.
\end{theorem}

For a matroid constraint, we get almost the same bound except for a slightly worse dependency on the parameter $\eta$:


\begin{theorem} \label{thm:main-submod-matroid}
For any $0 < \eps, \beta, \eta < 1$, there is an $\eps$-DP algorithm for monotone submodular maximization under matroid submodular maximization that achieves $\left(1 - \frac{1}{e} - \eta\right)$-approximation and $O\left(\frac{k \log\left(\frac{m}{\eta\beta}\right)}{\eta \eps}\right)$-error w.p. $1 - \beta$.
\end{theorem}

\textbf{Set Cover.}
For the private set cover problem, we give a pure-DP algorithm that improves the approximation ratio from \citet{GLMRT10} by a factor of $O(\log(1/\delta))$. This is tight due to the aforementioned $\Omega\left(\frac{\log m}{\eps}\right)$ from~\citet{GLMRT10} and the NP-hardness of factor $\Omega(\log n)$~\cite{DinurS14,Moshkovitz15}.

\begin{theorem} \label{thm:set-cover-main}
For any $0 < \eps, \beta < 1$, there is an $\eps$-DP algorithm for Set Cover that achieves $O\left(\log n + \frac{\log(m/\beta)}{\eps}\right)$-approximation w.p. $1 - \beta$.
\end{theorem}

\textbf{\boldmath Metric $k$-Means and $k$-Median}
For private metric $k$-means and $k$-median, we provide a pure-DP algorithm with approximation ratio $O(1)$ and error $O\left(\frac{k \log(mn)}{\eps}\right)$, as stated below. The error bound improves upon the approximate-DP algorithm of \citet{JNN21} by a factor of $O(\log(1/\delta))$. For $n \leq m^{O(1)}$, our error bound is tight due to the aforementioned lower bound $\Omega\left(\frac{k \log m}{\eps}\right)$ from \citet{GLMRT10}. Similar to \citep{JNN21}, we choose to keep our analysis simple, and thus we do not compute the approximation ratio explicitly. As we mentioned earlier, the tight approximation ratio is not known even in the non-private setting.

\begin{theorem} \label{thm:clustering}
For any $0 < \eps, \beta < 1$, there is an $\eps$-DP algorithm for metric $k$-median and $k$-means that achieves an $O(1)$-approximation and $O\left(\frac{k \log(mn/\beta)}{\eps}\right)$-error w.p. $1 - \beta$.
\end{theorem}

\textbf{Shifting Heavy Hitters.}
We provide a pure-DP algorithm for the problem, which is stated informally below; for a formal version, see \Cref{thm:shifting-hh}.
\begin{theorem}[Informal]
Assuming that each user contributes to $\leq k$ heavy hitters, there is an $\eps$-DP shifting heavy hitter algorithm with error $O\left(\frac{k \log(T|\cH|/\beta)}{\eps}\right)$ w.p. $1 - \beta$.
\end{theorem}
In comparison, the error bound of  \citet{KMS21} is $\tO\left(\frac{\sqrt{k} \log(1/\delta) \log(T|\cH|/\beta)}{\eps}\right)$, which can be smaller for large $k$. However, their algorithm is approximate-DP and it can be easily seen (via a packing lower bound) that our bound is the best possible for pure-DP; see \Cref{app:hh-lb}.

\subsection{Technical Overview}

There are two slightly different settings to which our techniques apply: \emph{repeated exponential mechanism} and \emph{repeated above threshold}. In this overview, we will focus on the repeated exponential mechanism and only briefly mention the repeated above threshold mechanism at the end.

\textbf{Repeated Exponential Mechanism.} \citet{GLMRT10} proposed the following algorithm for Set Cover: repeatedly use the $\eps_0$-DP exponential mechanism~\cite{MT07}\footnote{See \Cref{sec:prelim} for more detail.} to find the next set that covers the maximum number of (uncovered) element. Since the exponential mechanism is applied $m$ times here, if we were to apply a composition theorem, the error would grow polynomially in $m$ (either $m$ for basic composition or $\sqrt{m}$ for advanced composition~\cite{DworkRV10}). Perhaps surprisingly, they instead show that this algorithm is $(\eps, \delta)$-DP for $\eps = O(\eps_0 \cdot \log(1/\delta))$, i.e., independent of $m$.

The intuition behind this is roughly that an element ``causes'' only a single set to be picked: the first one in the permutation that contains it. The sets picked before this set  have their scores (of the exponential mechanism) completely independent from the element. On the other hand, for the sets picked after this set, the element is already covered and does not factor into the scores at all. Thus, we should be able to ``charge the privacy budget'' only once when this particular set is picked. While \citet{GLMRT10} show that such an \emph{individualized privacy accounting} works for approximate-DP, it unfortunately fails for pure-DP: this mechanism is not $\eps$-DP for any $\eps = o(m\eps_0)$. (See \Cref{app:counterexample}.)

\textbf{One-Sided DP.} This is a notion of DP where the ``neighboring relationship'' can be asymmetric~\cite{KDHM20}. (See also~\citep{TKCY22}, who call this \emph{asymmetric DP}.) In particular, we consider the following notion of one-sided DP with respect to adding an element\footnote{\citet{KDHM20} actually define one-sided DP with respect to \emph{replacing} a sensitive record. However, we are defining it with respect to adding an element.}. 

\begin{definition}[One-Sided DP,~\citet{KDHM20}]
A mechanism $\cM$ is said to be \emph{$\eps$-add-DP} iff, for every pair $D, D'$ of datasets such that $D$ results from adding an item to $D'$ and every possible subset $S$ of outputs, we have
$\Pr[\cM(D) \in S] \leq e^{\eps} \cdot \Pr[\cM(D') \in S]$.
\end{definition}

To emphasize the differences, we will refer to add-remove DP (\Cref{def:dp}) as ``two-sided DP''.

It turns out that, while it fails for two-sided DP, the above intuition does apply to one-sided DP: the repeated exponential mechanism is $\eps_0$-add-DP. The proof of this fact is also relatively straightforward (see \Cref{sec:rem}).

\textbf{From One-Sided to Two-Sided DP.} While the above observation is nice, we have not accomplished our goal yet, since we wish to design a two-sided DP mechanism, not a one-sided DP one. This brings us to our second observation: we can turn any one-sided DP mechanism into a two-sided one by subsampling with probability $p = 1 - e^{-\eps}$ (see \Cref{lem:subsample}). With these two ingredients, we arrive at our result by just using the subsampled version of the existing---e.g., \citet{GLMRT10}'s---algorithm!

\textbf{Repeated Above Threshold.} Our technique also applies to a slightly different setting where we consider a stream and wish to detect if a query is above a certain threshold at each time step. Again, we achieve an ``individualized privacy accounting'', where the amount of noise required to achieve one-sided DP scales only with the number of times an individual contributes to above threshold results. The subsampled version of this then satisfies two-sided DP and provides our algorithm for the shifting heavy hitter problem. 

Lastly, we note that, while we use the repeated exponential mechanism for submodular maximization results, we actually do \emph{not} use it for Set Cover. This is due to a technical barrier that only allows us to get $O(\log n \log m / \eps)$ via this approach. (See \Cref{app:set-cov-non-opt} for more details.) For \Cref{thm:set-cover-main}, we actually use the repeated above threshold mechanism applied to a (non-private) streaming approximation algorithm for Set Cover~\cite{KMVV15}.

\section{Preliminaries}\label{sec:prelim}


\textbf{Subsampling.}
Subsampling is a standard technique in DP~\cite{BBG18,WBK20}. We will use the so-called Poisson subsampling where each user is kept with probability $p$. More precisely, we write $\cS^p: \cX^* \to \cX^*$ as a subsampling operator, i.e., $\cS_p(D)$ outputs a random subset of $D$ such that each user is kept independently with probability $p$. For any mechanism $\cM$, we write $\cM^{\cS_p}$ to denote the mechanism $D \mapsto \cM(\cS_p(D))$.

\textbf{Concentration Inequalities.}
It will be convenient to use a version of the Chernoff bound that includes both multiplicative and additive terms, as stated below.

\begin{theorem}[Chernoff bound; Corollary 2.11 of \citealt{CNZ-arxiv})] \label{thm:chernoff}
Let $Z_1, \dots, Z_m$ be i.i.d. random variables such that $Z_i \in [0, 1]$ and let $\mu = \E[Z_1 + \cdots + Z_m]$. Then, for $\alpha \in [0, 1]$ and $\zeta \geq 0$, we have
\begin{align*}
\Pr[Z_1 + \cdots + Z_m < (1 - \alpha)\mu - \zeta] &\leq \exp\left(-\alpha\zeta\right), \\
\Pr[Z_1 + \cdots + Z_m > (1 + \alpha)\mu + \zeta] &\leq \exp\left(-\alpha\zeta/3\right). 
\end{align*}
\end{theorem}
\textbf{Exponential Mechanism.}
Given a candidate set $\cC$ and a scoring function $q: \cC \times \cX^* \to \R$, the \emph{exponential mechanism} $\EM_{\eps}(\cC, q; D)$ outputs  each candidate $c \in \cC$ with probability proportional to $\exp\left(\eps \cdot q(c; D) \right)$. Its guarantee is as follows:
\begin{theorem}[\citet{MT07}] \label{thm:em}
If $q$ has sensitivity at most $\Delta$ (w.r.t. $D$), then the exponential mechanism is $(2\eps\Delta)$-DP. Furthermore, w.p. $1 - \beta$, the output $c^*$ satisfies $q(c^*; D) \geq \max_{c \in \cC} q(c; D) - O(\log(|\cC|/\beta)/\eps)$.
\end{theorem}

\section{One-Sided DP and Subsampling}

We start by proving our main observation that subsampling a one-sided DP mechanism makes it two-sided DP. We remark that, while there is a large literature on privacy amplification by subsampling, we are not aware of such a connection between one-sided and two-sided DP before.

\begin{lemma} \label{lem:subsample}
For any $p \in [0, 1)$ and $\eps_0 > 0$, if $\cM$ is an $\eps_0$-add-DP mechanism, then $\cM^{\cS_p}$ is $\eps$-DP for $\eps = \ln\left(\max\left\{\frac{1}{1 - p}, 1 + p(e^{\eps_0} - 1)\right\}\right)$.
\end{lemma}
\begin{proof}
Consider any datasets $D, D'$ such that $D = D' \cup \{x\}$, and any possible output $o$. On the one hand, we have
\begin{align*}
&\Pr[\cM^{\cS_p}(D) = o] \\
&= \sum_{D_{\subs} \subseteq D}  \Pr[\cM(D_{\subs}) = o] \cdot \Pr[\cS_p(D) = D_{\subs}] \\
&\geq \sum_{D_{\subs} \subseteq D'} \Pr[\cM(D_{\subs}) = o] \cdot \Pr[\cS_p(D) = D_{\subs}] \\
&\overset{(\star)}{=} \sum_{D_{\subs} \subseteq D'} \Pr[\cM(D_{\subs}) = o] \cdot(1 - p) \Pr[\cS_p(D') = D_{\subs}] \\
&= (1 - p) \cdot \Pr[\cM^{\cS_p}(D') = o],
\end{align*}
where $(\star)$ uses the fact that $\cS_p(D)$ is the same as $\cS_p(D')$ when conditioned on $x$ not being selected. 

On the other hand, we have\footnote{The following sequence of inequalities is standard in amplification-by-subsampling literature (e.g.,~\cite{LQS12}); we only include it here for completeness.}
\begin{align*}
&\Pr[\cM^{\cS_p}(D) = o] \\
&= \sum_{D_{\subs} \subseteq D'} \bigg(\Pr[\cM(D_{\subs}) = o] \cdot \Pr[\cS_p(D) = D_{\subs}] \\
&\qquad + \Pr[\cM(D_{\subs} \cup \{x\}) = o] \cdot \Pr[\cS_p(D) = D_{\subs} \cup \{x\}]\bigg) \\
&\overset{(\bigstar)}{\leq} \sum_{D_{\subs} \subseteq D'} \bigg(\Pr[\cM(D_{\subs}) = o] \cdot \Pr[\cS_p(D) = D_{\subs}] \\
&\qquad + e^{\eps_0} \cdot \Pr[\cM(D_{\subs}) = o] \cdot \Pr[\cS_p(D) = D_{\subs} \cup \{x\}]\bigg) \\
&\overset{(\blacklozenge)}{=} \sum_{D_{\subs} \subseteq D'} \bigg(\Pr[\cM(D_{\subs}) = o]\cdot (1-p)\cdot\Pr[\cS_p(D') = D_{\subs}] \\
&\qquad + e^{\eps_0} \cdot \Pr[\cM(D_{\subs}) = o] \cdot p\cdot\Pr[\cS_p(D') = D_{\subs}]\bigg) \\
 &= \left(1 + p(e^{\eps_0} - 1)\right) \cdot \Pr[\cM^{\cS_p}(D') = o],
\end{align*}
where $\bigstar$ follows from the $\eps_0$-add-DP property of $\cM$ and $\blacklozenge$ follows from the fact that $x$ is included in $D_{\subs}$ with probability $p$ independently of other items.

Thus, the algorithm is $\eps$-DP as claimed.
\end{proof}

We note that \Cref{lem:subsample} can be extended to approximate-DP or R\'enyi-DP by replacing the second half of the proof by the corresponding amplification-by-subsampling proofs.

The following corollary, which is an immediate consequence of plugging in $\eps_0 = \ln(2)$ and $p = 1 - e^{-\eps}$ into~\Cref{lem:subsample}, will be more convenient to work with throughout the remainder of the paper. It is useful to note that $p = \Theta(\eps)$ for $\eps \leq 1$ while $\eps_0 = \ln(2)$ is independent of $\eps$.

\begin{corollary} \label{cor:subsampling}
For any $\eps > 0$, if $\cM$ is $\ln(2)$-add-DP, then $\cM^{\cS_p}$ is $\eps$-DP for $p = 1 - e^{-\eps}$.
\end{corollary}

\section{Algorithm I: Repeated EM}
\label{sec:rem}

In the first setting, the interaction proceeds in $L$ rounds. In round $i$, the algorithm is given a candidate set $\cC_i$ and a scoring function $q_i: \cC_i \times \cX^* \to \R$ (which can depend on the output of previous rounds). The goal is to output $c^*_i \in \cC_i$ which achieves an approximately maximum score $q_i(c; D)$. The algorithm we use (\Cref{alg:repeated-em})---originally from \citet{GLMRT10}---simply applies the exponential mechanism at each step.

To analyze the algorithm, we need a couple of assumptions.
\begin{assumption}[Monotonicity] \label{assum:monotone}
For every $i, c$, adding an element to $D$ does not decrease $q_i(c; D)$.
\end{assumption}

\begin{assumption}[Bounded Realized Sensitivity] \label{assum:marginal-sensitivity}
For every add-remove neighbors $D, D'$ and every possible output $(c^*_1, \dots, c^*_L)$, $\sum_{i \in [L]} |q_i(c^*_i; D) - q_i(c^*_i; D')| \leq \Delta$.
\end{assumption}

\begin{algorithm}[t]
\caption{$\rem_{\eps_0, \cA}$~\cite{GLMRT10}}
\label{alg:repeated-em}
\begin{algorithmic}
\PARAMETERS $\eps_0 > 0$, an algorithm $\cA$ for selecting a candidate set and a scoring function
\REQUIRE Dataset $D \in \cX^*$
\FOR{$i = 1, \dots, L$}
  \STATE $(\cC_i, q_i) \gets \cA(c^*_1, \dots, c^*_{i - 1})$
  \STATE $c^*_i \gets \EM_{\frac{\eps_0}{\Delta}}(\cC_i, q_i; D)$
\ENDFOR
\RETURN $(c^*_1, \dots, c^*_L)$
\end{algorithmic}
\end{algorithm}

Under these assumptions, the algorithm is add-DP:

\begin{theorem} \label{thm:rem-add-dp}
Under Assumptions~\ref{assum:monotone}~and~\ref{assum:marginal-sensitivity},  {\sc Repeated-EM$_{\eps_0}$} (\Cref{alg:repeated-em}) is $\eps_0$-add-DP.
\end{theorem}

\begin{proof}
Consider any $D, D'$ such that $D = D' \cup \{x\}$ for some $x$, and any output $o = (o_1, \dots, o_L)$. We will write $\cM$ as a shorthand for $\rem_{\eps_0, \cA}$. We have
\begin{align*}
\frac{\Pr[\cM(D) = o]}{\Pr[\cM(D') = o]}
&= \prod_{i \in [L]} \frac{\Pr[\EM_{\frac{\eps_0}{\Delta}}(\cC_i, q_i; D) = o_i]}{\Pr[\EM_{\frac{\eps_0}{\Delta}}(\cC_i, q_i; D') = o_i]} \\
&= \prod_{i \in [L]} \frac{\frac{\exp\left(\frac{\eps_0}{\Delta} \cdot q_i(o_i; D)\right)}{\sum_{c \in \cC_i} \exp\left(\frac{\eps_0}{\Delta} \cdot q_i(c; D)\right)}}{\frac{\exp\left(\frac{\eps_0}{\Delta} \cdot q_i(o_i; D')\right)}{\sum_{c \in \cC_i} \exp\left(\frac{\eps_0}{\Delta} \cdot q_i(c; D')\right)}}.
\end{align*}
\Cref{assum:monotone} implies that $\sum_{c \in \cC_i} \exp\left(\frac{\eps_0}{\Delta} \cdot q_i(c; D)\right) \geq \sum_{c \in \cC_i} \exp\left(\frac{\eps_0}{\Delta} \cdot q_i(c; D')\right)$. Thus, we have 
\begin{align*}
\frac{\Pr[\cM(D) = o]}{\Pr[\cM(D') = o]} &\leq \exp\left(\frac{\eps_0}{\Delta} \cdot \sum_{i \in [L]} (q_i(o_i; D) - q_i(o_i; D')) \right) \\ &\leq \exp(\eps_0),
\end{align*}
where the second inequality follows from \Cref{assum:marginal-sensitivity}.

As a result, $\cM$ is $\eps_0$-add-DP, concluding our proof.
\end{proof}

By applying \Cref{cor:subsampling}, we immediately have that its subsampled variant is $\eps$-DP:

\begin{theorem} \label{thm:subsample-dp}
Under Assumptions~\ref{assum:monotone}~and~\ref{assum:marginal-sensitivity}, $\rem_{\ln(2), \cA}^{\cS_p}$ is $\eps$-DP for $p = 1 - e^{-\eps}$.
\end{theorem}

\subsection{Applications}

\subsubsection{Monotone Submodular Maximization under Cardinality Constraint} \label{sec:smc}

The algorithm in~\citet{GLMRT10} for monotone submodular maximization under cardinality constraint is based on the classic greedy algorithm that runs in $k$ rounds, each round finding an element that maximizes the marginal gain. More precisely, the algorithm---which we  call $\smcalg_{\eps_0, F_D}$---is exactly $\rem_{\eps_0, \cA}$ where $L = k$ and the candidate sets and scoring functions are as follows:
\begin{itemize}[nosep]
\item $\cC_i$ is the set of remaining elements $[m] \setminus \{c^*_1, \dots, c^*_{i - 1}\}$ 
\item $q_i(c; D)$ is the marginal gain $F_D(\{c^*_1, \dots, c^*_{i - 1}, c\}) - F_D(\{c^*_1, \dots, c^*_{i - 1}\})$
\end{itemize}

They proved the following utility guarantee:

\begin{theorem}[\citealt{GLMRT10}] \label{thm:smc-gupta-util}
For any $\eps_0, \beta > 0$, $\smcalg_{\eps_0}$ achieves $(1 - 1/e)$-approximation and $O\left(\frac{k \log (m/\beta)}{\eps_0}\right)$-error with probability $1 - \beta$.
\end{theorem}
%
%
\begin{proof}[Proof of \Cref{thm:smc}]
We simply run the subsampled version of \Cref{alg:repeated-em}. More precisely, we use the algorithm $\smcalg_{\ln(2), F_D}^{\cS_p}$ where $p = 1 - e^{-\eps}$. The privacy analysis follows from the straightforward observation that $\cC_i, q_i$ in $\smcalg$ satisfies Assumptions~\ref{assum:monotone} and~\ref{assum:marginal-sensitivity} and \Cref{thm:subsample-dp}. 

For the utility, let $D_{\subs} \sim \cS_p(D)$ denote the subsampled dataset that is fed as an input to $\smcalg_{\ln(2)}$ and let $T^* := \{c^*_1, \dots, c^*_k\}$ denote the output set. 
From \Cref{thm:smc-gupta-util}, w.p. $1 - \beta/2$, we have
\begin{align} 
F_{D_{\subs}}(T^*) &\geq \left(1 - \frac{1}{e}\right) \cdot \max_{T \subseteq \cS, |T| = k} F_{D_{\subs}}(T) & \nonumber \\ &\qquad - O\left(k \log \left(\frac{m}{\beta}\right)\right). \label{eq:smc-util-subsampled}
\end{align}
Furthermore, applying the Chernoff bound (\Cref{thm:chernoff}) with $Z_x := f_x(T) \cdot \ind[x \in D_{\subs}], \mu = p \cdot F_D(T), \alpha = 0.01\eta, \zeta = \frac{2000 k \log(m/\beta)}{\eta}$ and a union bound over all sets $T \in \binom{\cU}{\leq k}$, we can conclude that the following hold simultaneously for all $T \in \binom{\cU}{\leq k}$ w.p. at least $1 - \beta/2$:
\begin{align}
F_{D_{\subs}}(T) &\geq (1 - \alpha)p \cdot F_D(T) - \zeta \label{eq:smc-lb}, \\
F_{D_{\subs}}(T) &\leq (1 + \alpha)p \cdot F_D(T) + \zeta \label{eq:smc-ub}.
\end{align}
When \eqref{eq:smc-util-subsampled}, \eqref{eq:smc-lb}, and \eqref{eq:smc-ub} all hold, we have
\begin{align*}
&F_D(T^*)  \\
&\overset{\eqref{eq:smc-ub}}{\geq} \frac{1}{(1 + \alpha)p} F_{D_{\subs}}(T^*) - \zeta / p \\
&\overset{\eqref{eq:smc-util-subsampled}}{\geq} \frac{1}{(1 + \alpha)p} \cdot \left(1 - \frac{1}{e}\right) \cdot \max_{T \subseteq \cS, |T| = k} F_{D_{\subs}}(T) \\&\qquad- O\left(\frac{k \log (m/\beta)}{p}\right) - \zeta / p \\
&\overset{\eqref{eq:smc-lb}}{\geq} \frac{1 - \alpha}{1 + \alpha} \cdot \left(1 - \frac{1}{e}\right) \cdot \max_{T \subseteq \cS, |T| = k} F_{D}(T) \\&\qquad - O\left(\frac{k \log (m/\beta)}{p}\right) - 2\zeta / p \\
&\geq \left(1 - \frac{1}{e} - \eta\right) \cdot \max_{T \subseteq \cS, |T| = k} F_{D}(T) - O\left(\frac{k \log (m/\beta)}{\eta \eps}\right),
\end{align*}
which concludes our proof.
\end{proof}

\subsubsection{Monotone Submodular Maximization under Matroid Constraint}

For maximization over a matroid, the greedy algorithm is \emph{not} known to achieve $1 - 1/e$ approximation ratio. Instead, one has to resort to the so-called \emph{continuous greedy algorithm} of \citet{CCPV11} (which is in turn based on an earlier algorithmic idea by \citet{Vondrak08}). \citet{CNZ21} followed this route and privatized the continuous greedy algorithm, albeit only achieving approximate-DP. Similarly to the above, we show that this algorithm in fact satisfies one-sided DP and, using the subsampled version of it, we prove \Cref{thm:main-submod-matroid}. Due to space constraints, we defer the full proof to~\Cref{app:submod-matroid}.

\subsubsection{Metric $k$-Means and $k$-Median}

For metric $k$-means/median, we first use the repeated exponential mechanism to select $O(k \log n)$ points that form a \emph{bicriteria} $O(1)$-approximate solution (where ``bicriteria'' refers to the fact that the set has size larger than $k$). To turn a bicriteria solution to an actual solution, we use a standard technique in private clustering~\cite{JNN21,GKM20}: we snap each input point to the closest point in the bicriteria solution and add noise to the counts to create ``private synthetic data''. We can then run any non-private approximation algorithm on this synthetic data to get our ultimate solution. As described, this algorithm only satisfies one-sided DP. To achieve two-sided DP, we again use the subsampled version of it, 
similar to \Cref{thm:smc}. The full proof is deferred to~\Cref{app:clustering}.

\section{Algorithm II: Repeated Above Threshold}

Again, the interaction proceeds in $L$ rounds. In round $i$, the algorithm is given a function $h_i: \cX^* \to \R$ together with a threshold $\tau_i$ (which can depend on the outputs from the previous rounds). The goal is to decide whether $h_i(D) \geq \tau_i$. The algorithm we consider is one that repeatedly applies a variant of the AboveThreshold algorithm \cite{DNRRV09}\footnote{See also Algorithm 1 in \citet{DworkR14}.}, where we only add noise to the function (but not to the threshold as in \cite{DNRRV09}) and the noise is drawn from the exponential distribution\footnote{Recall that $\Exp(\lambda)$ has CDF $1 - e^{-\lambda x}$ for $x \in [0, \infty)$.} (rather than the Laplace distribution); see \Cref{alg:repeated-abovethreshold}. We remark that our algorithm is also different from both previous works of~\citet{KMS21} (who need add another noise term to make their analysis work) and of~\citet{CL23} (who simply use the Laplace mechanism in each round). 

We analyze the DP guarantee of \Cref{alg:repeated-abovethreshold} under the following assumptions (similar to the ones in \Cref{sec:rem}).

\begin{assumption}[Monotonicity] \label{assum:monotone-at}
For every $i$, adding an element to $D$ does not decrease $h_i(D)$.
\end{assumption}

\begin{assumption}[Bounded Realized Sensitivity] \label{assum:marginal-sensitivity-at}
For every add-remove neighbors $D, D'$ and all possible sequences  $(h_1, \dots, h_L)$ of functions  and outputs $(o_1, \dots, o_L)$, $\sum_{i \in [L] \atop o_i = \top} |h_i(D) - h_i(D')| \leq \Delta$.
\end{assumption}

\begin{algorithm}
\caption{$\rat_{\eps_0, \cA}$}
\label{alg:repeated-abovethreshold}
\begin{algorithmic}    
\PARAMETERS $\eps_0 > 0$, an algorithm $\cA$ for selecting candidate and scoring functions
\REQUIRE Dataset $D \in \cX^*$
\FOR{$i = 1, \dots, L$}
\STATE $(h_i, \tau_i) \gets \cA(c^*_1, \dots, c^*_{i - 1})$
\STATE $\theta_i \sim \Exp(\eps_0 / \Delta)$
\IF{$h_i(D) + \theta_i > \tau_i$}
\STATE $c^*_i = \top$
\ELSE
\STATE $c^*_i = \perp$
\ENDIF
\ENDFOR
\RETURN $(c^*_1, \dots, c^*_L)$
\end{algorithmic}
\end{algorithm}

\begin{theorem}
Under Assumptions~\ref{assum:monotone-at}~and~\ref{assum:marginal-sensitivity-at}, $\rat_{\eps_0, \cA}$ (\Cref{alg:repeated-abovethreshold}) is $\eps_0$-add-DP.
\end{theorem}

\begin{proof}
Consider any $D, D'$ such that $D = D' \cup \{x\}$ for some $x$, and any output $o = (o_1, \dots, o_L)$. We will write $\cM$ as a shorthand for $\rat_{\eps_0, \cA}$. Let $\cI_{\top} := \{i \in L \mid o_i = \top\}$ and $\cI_{\bot} := \{i \in L \mid o_i = \bot\}$. We have
\begin{align*}
\frac{\Pr[\cM(D) = o]}{\Pr[\cM(D') = o]}
&= \prod_{i \in \cI_{\top}} \frac{\Pr[h_i(D) + \theta_i > \tau_i]}{\Pr[h_i(D') + \theta_i > \tau_i]} \\ &\qquad \cdot \prod_{i \in \cI_{\bot}} \frac{\Pr[h_i(D) + \theta_i \leq \tau_i]}{\Pr[h_i(D') + \theta_i \leq \tau_i]}.
\end{align*}
Since $h_i(D') \leq h_i(D)$ (\Cref{assum:monotone-at}), we have $\frac{\Pr[h_i(D) + \theta_i \leq \tau_i]}{\Pr[h_i(D') + \theta_i \leq \tau_i]} \leq 1$. Meanwhile, from the definition of $\Exp\left(\frac{\eps_0}{\Delta}\right)$, we have $\frac{\Pr[h_i(D) + \theta_i \leq \tau_i]}{\Pr[h_i(D') + \theta_i \leq \tau_i]} \leq \exp\left(\frac{\eps_0}{\Delta}\cdot(h_i(D) - h_i(D'))\right)$. 
Thus, we have
\begin{align*}
\frac{\Pr[\cM(D) = o]}{\Pr[\cM(D') = o]} &\leq \exp\left(\frac{\eps_0}{\Delta} \cdot \sum_{i \in \cI_{\top}} (h_i(D) - h_i(D'))\right) \\ &\leq \exp(\eps_0),
\end{align*}
where the second inequality is from \Cref{assum:marginal-sensitivity-at}.  
\end{proof}

By applying \Cref{cor:subsampling}, we immediately have that its subsampled variant is $\eps$-DP:

\begin{theorem} \label{thm:subsample-dp-rat}
Under Assumptions~\ref{assum:monotone-at}~and~\ref{assum:marginal-sensitivity-at}, $\rat_{\ln(2), \cA}^{\cS_p}$ is $\eps$-DP for $p = 1 - e^{-\eps}$.
\end{theorem}

\subsection{Applications}

\subsubsection{Shifting Heavy Hitters}

To formalize our result, we need to first formalize the assumption that ``each user contributes to at most $k$ heavy hitters''. To do this, let us first define $\tau^*(k) := \frac{4000 k \log(T|\cH|/\beta)}{\eps}$ to be (half of) our target error.
The assumption we work with is the following, which is the same as that of \citet{KMS21} (but with different $\tau^*(k)$).

\begin{assumption} \label{assum:shifting-hh}
For every user $i \in [n]$, we have \\ $|\{t \in [T] \mid w_t(x_i) > \tau^*(k)\}| \leq k$.
\end{assumption}

Our theorem can now be formalized as follows.

\begin{theorem} \label{thm:shifting-hh}
For any $0 < \eps \leq 1$, under \Cref{assum:shifting-hh} for $\tau^*(k) = \frac{1000 k \log(T|\cH|/\beta)}{\eps}$, there is a shifting heavy hitters algorithm with error $2\tau^*(k)$ w.p. $1 - \beta$.
\end{theorem}
Note that \Cref{assum:shifting-hh} is required only for utility; privacy guarantee holds for all input datasets, as is standard in DP.

Our $\threshmon_{\eps_0, \tau, k}$ algorithm (which is a simplification of the algorithm from \citet{KMS21}) is presented in \Cref{alg:threshold-monitor}. Here, we keep the counter $C_i$ for the number of times the user $i$ has contributed to the heavy hitters. When this hits $k$, we simply drop this user and never include this user in the counts in the subsequent rounds.

\begin{algorithm}
\caption{$\threshmon_{\eps_0, \tau, k}$}
\label{alg:threshold-monitor}
\begin{algorithmic}
\PARAMETERS $\eps_0 > 0$, $\tau$ the desired heavy hitter threshold, $k \in \N$ limit on the number of contributions
\REQUIRE Input data stream $D \in (\cH^T)^n$
\STATE $\cI \gets [n]$
\FOR{$i = 1, \dots, n$}
\STATE $C_i \gets 0$
\ENDFOR
\FOR{$t = 1, \dots, T$}
\FOR{$y \in \cH$}
\STATE $w_t^{\cI}(y; D) \gets |\{i \in \cI \mid x_{i, t} = y\}|$
\STATE $\theta_{t, y} \sim \Exp(\eps_0 / k)$
\IF{$w_t^{\cI}(y; D) + \theta_{t, y} > \tau_i$}
\STATE Report $y$ for time step $t$
\FOR{$i \in \cI$ such that $x_{i, t} = y$}
\STATE $C_i \gets C_i + 1$
\IF{$C_i = k$}
\STATE $\cI \gets \cI \setminus \{i\}$
\ENDIF
\ENDFOR
\ENDIF
\ENDFOR
\ENDFOR
\end{algorithmic}
\end{algorithm}

We are now ready to prove \Cref{thm:shifting-hh}.

\begin{proof}[Proof of \Cref{thm:shifting-hh}]
We use the subsampled version of $\threshmon$, i.e.,  $\threshmon_{\ln(2), \tau, k}^{\cS_p}$ for $p = 1 - e^{-\eps}$ and $\tau = 1.5 p \cdot \tau^*(k)$.
It is not hard to see that $\threshmon$ is an instantiation of \Cref{alg:repeated-abovethreshold} with $h_{ t, x}^{\cI}(D) := w_{t}^{\cI}(x)$ that satisfies \Cref{assum:monotone-at} and \Cref{assum:marginal-sensitivity-at} with $\Delta = k$.
 Thus, \Cref{thm:subsample-dp-rat} implies that $\threshmon_{\ln(2), \tau, k}^{\cS_p}$ is $\eps$-DP as desired.

To see the utility guarantee, let $D_{\subs}$ denote the subsampled dataset used as the input to $\threshmon_{\ln(2), \tau, k}$. Note also that $\tau^*(k) \cdot p \geq 0.5 \tau^*(k) \cdot \eps = 2000 k \ln(T|\cH|/\beta)$.  We first apply the Chernoff bound (\Cref{thm:chernoff}) with $Z_i := \ind\left[i \in D_{\subs}\right] \cdot \ind[x_{i,t} = y], \mu = p \cdot w_t(y; D), \alpha = 0.1, \zeta = 100 \cdot \ln(2T|\cH|/\beta) \leq 0.1p \cdot \tau^*(k)$ together with a union bound over all $t \in [T], y \in \cH$, we can conclude that the following hold simultaneously for all  $t \in [T], y \in \cH$ with probability at least $1 - \beta/2$:
\begin{align} 
w_t(y; D_{\subs}) &\geq 0.9 p \cdot w_t(y; D) - 0.1p \cdot \tau^*(k), \label{eq:hh-lb} \\
w_t(y; D_{\subs}) &\leq 1.1 p \cdot w_t(y; D) + 0.1p \cdot \tau^*(k). \label{eq:hh-ub}
\end{align}

Furthermore, by tail bounds for exponential noise, the following holds for all $t \in [T], x \in \cX$ with probability $1 - \beta / 2$:
\begin{align} \label{eq:noise-bound-hh}
\theta_{t,y} \leq k \log(2T|\cH|/\beta) \leq 0.01\tau^*(k).
\end{align}

We will continue the remainder of the analysis assuming that  \eqref{eq:hh-lb}, \eqref{eq:hh-ub}, and \eqref{eq:noise-bound-hh} hold for all $t \in [T], y \in \cH$; by a union bound, this occurs with probability at least $1 - \beta$. 

By \eqref{eq:hh-ub} and \eqref{eq:noise-bound-hh}, if $y$ is reported at time  $t$, then we must have
\begin{align*}
w_t(y; D) \geq \frac{\tau - 0.1p \cdot \tau^*(k)}{1.1 p} > \tau^*(k).
\end{align*}
This satisfies the first part of the accuracy requirement. Furthermore, this and \Cref{assum:shifting-hh} imply that each data record $i$ in $D_{\subs}$ contributes to at most $k$ reported heavy hitters. Thus, $\cI$ remains the entire dataset for the entire run. As a result, for each $(2\tau^*(k))$-heavy hitter $y$ at time step $t$,
\begin{align*}
w_t^{\cI}(y; D) \overset{\eqref{eq:hh-lb}}{\geq} 0.9 p \cdot 2\tau^*(k) - 0.1 p \cdot \tau^*(k) > \tau.
\end{align*}
Since $\theta_{t, y} \geq 0$, this implies that $y$ is reported at time $t$. Thus, the algorithm satisfies the claimed accuracy bounds.
\end{proof}

\subsubsection{Set Cover}
\label{sec:set-cov}

We use the so-called \emph{greedy scaling} algorithm from~\citet{KMVV15}. Their algorithm works in $O(\log n)$ rounds. In each round, we iteratively keep all sets that cover at least a certain number of elements; this threshold is geometrically decreased across different rounds until every element is covered. We adapt this algorithm to the DP setting using our $\rat$ algorithm (\Cref{alg:repeated-abovethreshold}) to find the sets to be picked in each round. Note that we also resample the dataset $D_{\subs}$ in each round to avoid having to do a union bound over too large a number of events. This requires us to carefully split the privacy budget in each round (which we assigns in geometrically increasing manners). The complete description is presented in \Cref{alg:set-cov}. Due to space constraints, we defer the full proof to \Cref{app:set-cov}.

\begin{algorithm}
\caption{$\setcovalg_{\eps}$}
\label{alg:set-cov}
\begin{algorithmic}
\PARAMETERS $\eps > 0$
\REQUIRE Input universe $D \in \cX^*$, subsets $S_1, \dots, S_m \subseteq D$
\STATE $\cI \gets [m]$
\STATE $n_{\min} \gets 100 \log m$
\STATE $\tn \gets \max\left\{n + \Lap(0.5/\eps), n_{\min}\right\}$
\STATE $j \gets 1$
\STATE $R \gets \lfloor \log(\tn / n_{\min}) \rfloor$
\FOR{$r = 1, \dots, R$}
\STATE $\tau_r \gets 1000 \cdot \frac{\tn}{2^r}$
\STATE $\eps_r \gets \frac{\eps}{4 \cdot 2^{R-r}}$
\STATE $p_r \gets 1 - e^{-\eps_r}$
\STATE $D_{\subs, r} \sim \cS^{p_r}(D)$
\FOR{$i \in \cI$}
\STATE $\theta_{r, i} \sim \Exp(\ln(2))$
\IF{$\left|(S_i \cap D_{\subs, r}) \setminus \left(\bigcup_{j' < j} S_{\pi(j')} \right)\right| + \theta_{r, i} > p_r \cdot \tau_r$}
\STATE $\pi(j) \gets i$
\STATE $\cI \gets \cI \setminus \{i\}$
\STATE $j \gets j + 1$
\ENDIF
\ENDFOR
\ENDFOR
\FOR{$i \in \cI$}
\STATE $\pi(j) \gets i$
\STATE $j \gets j + 1$
\ENDFOR
\RETURN $(\pi(1), \dots, \pi(m))$
\end{algorithmic}
\end{algorithm}

\section{Conclusion and Open Questions}
\label{sec:conc}

In this work, we make a simple observation that subsampling a one-sided DP mechanism makes it two-sided DP. Applying this observation to the repeated exponential mechanism and the repeated above threshold mechanism, we obtain novel pure-DP algorithms for several combinatorial  optimization problems and for the shifting heavy hitters problem. It remains interesting to explore the applications of this framework further. One clear barrier of the current approach is that it requires monotonicity (Assumptions \ref{assum:monotone} and~\ref{assum:monotone-at}). This prevents us from applying this to the non-monotone submodular maximization problems; meanwhile \cite{CNZ21} show that a Gupta et al.-like analysis still works for approximate-DP. In particular, they achieve $\left(\frac{1}{e} - \eta\right)$-approximation and $O\left(\frac{k \log\left(\frac{m}{\eta\beta}\right)\log(1/\delta)}{\eta \eps}\right)$-error for non-monotone submodular maximization under matroid constraint. A concrete question here is whether we can achieve a similar guarantee under pure-DP. 

\bibliographystyle{icml2024}
\bibliography{ref}

\newpage
\onecolumn
\appendix 

\section{Additional Related Work}

Below we provide additional discussion on related work.

\paragraph{AboveThreshold with Individualized Privacy Loss.}
As stated earlier, \citet{KMS21} provide a modification of the sparse vector technique (SVT) \cite{DNRRV09} that can be used for individualized privacy loss. Their algorithm is similar to \Cref{alg:repeated-abovethreshold} presented in our work except that (i) they use Laplace noise (since they want two-sided DP), (ii) they add noise to the threshold (similar to standard SVT) and (iii) they also add another ``noise of noise'' term to the query value. The last one is to help with the analysis but results in an increase of $O(\log(\log(1/\delta)/\eps))$ in their error bound. Very recently,~\citet{CL23} gave a different algorithm for the task that gets rid of this bound. Their algorithm is essentially the same as our \Cref{alg:repeated-abovethreshold} with the exception that they use Laplace noise instead of Exponential noise. The framework of \citet{CL23} is extremely generic and individualized privacy loss accounting is only one of the applications of their framework. However, their results only apply for approximate-DP. It remains interesting to see if our framework can be used for any of the applications in their paper.

\paragraph{Private Submodular Maximization.} 
Although our paper assumes that the function is decomposable, DP submodular maximization has also been studied under other assumptions. For example,~\citet{MBKK17} also study the setting where the function $F_D$ is only assumed to have low-sensitivity, i.e., $|F_D(S) - F_{D'}(S)| \leq \Delta$ for all neighboring datasets $D, D'$. This is a more relaxed assumption than decomposability. Our techniques do not seem to apply here, both because (i) it is not clear how to related the subsampled function value to the function value on the entire dataset and (ii) the individualized privacy accounting is not known to be applicable here even for approximate-DP. It remains an interesting question whether one can get improved bound under this weaker assumption. Note that, for 1-sensitive monotone submodular functions, the best $\eps$-DP algorithm under cardinality constraint is still from \cite{MBKK17} and achieves approximation ratio $\left(1 - \frac{1}{e}\right)$ and error $O\left(\frac{k^2\log n}{\eps}\right)$. For matroid constraint, the best algorithm is that of \citet{RY20}, which achieves the same approximation ratio but with error $O\left(\frac{k^7 \log n}{\eps^3}\right)$. 



Another related work on the topic is by~\citet{SF21}, who gave an improved approximation algorithm if the \emph{total curvature} of the function is smaller than one and also proposed an algorithm for the online setting. A recent work of~\citet{CNN23} also studied private submodular maximization in the streaming setting. Finally, we note that the aforementioned work also considered other settings including non-monotone functions. Our techniques do not apply in this case; we provide more discussion regarding this in~\Cref{sec:conc}.

\paragraph{Partial Set Cover.} \citet{LNV23} initiated a study of private partial set cover, where it suffices to cover a fraction (not all) of the elements, and apply it as a subroutine to the DP $k$-supplier with outlier problem. Since they also use Gupta et al.'s algorithm, our technique can be applied to their setting to achieve improved bounds as well.

\paragraph{Private Clustering.} We remark that in \cite{GLMRT10,JNN21}, the metric $k$-median problem is defined in so that each point in the metric can appear in the dataset only once. This implies $n \leq m$ and their bound thus only depends on $m$. We choose to state the more general formulation and bounds here, which is why we also have the dependency on $n$.

While we focus our attention on discrete finite metric, DP clusterings have also been studied in other infinite metrics, such as the $\ell_p$-metric~\cite{BDMN05,FeldmanFKN09,NissimS18,StemmerK18,GKM20,JNN21,NCX21,ChangGKM21}. Some of these works, e.g.,~\cite{JNN21,NCX21}, uses the repeated exponential mechanism as a subroutine. Therefore, our techniques can be applied to reduce the errors in these bounds.

\paragraph{Amplification by Subsampling.} There is, by now, a large body of literature on DP amplification by subsampling (e.g.,~\cite{LQS12,BBG18,WBK20,dpsgd}). However, we are not aware of any result that allows us to goes from an approximate-DP guarantee to a pure-DP guarantee, which is what \Cref{thm:subsample-dp} achieves (albeit with a one-sided DP requirement). It remains an intriguing question to further explore the power of amplification by subsampling.

\section{Lower Bound for Shifting Heavy Hitters}
\label{app:hh-lb}

We prove the following lower bound for the shifting heavy hitters problem, which shows that the bound on $\tau^*(k)$ we achieved in \Cref{thm:shifting-hh} is essentially tight.

\begin{theorem}
For any sufficiently large $k, T \in \N$ such that $T \geq k, \eps \in (0, 1)$ and $|\cH| \geq 5$, there is no $\eps$-DP algorithm that can achieve $2\tau^*(k)$ error w.p. 0.5 under \Cref{assum:shifting-hh} for $\tau^*(k) = 0.001k \log\left(\frac{|\cH| T}{k}\right)$.
\end{theorem}

\begin{proof}
Assume w.l.o.g. that $\cH = [r + 3]$ for $r \in \N$.
Suppose for the sake of contradiction that there exists an $\eps$-DP algorithm $\cM$ that achieve $2\tau^*(k)$ error w.p. 0.5 under \Cref{assum:shifting-hh} for $\tau^*(k) = 0.001k \log\left(\frac{|\cH| T}{k}\right)$. Let $\cX'$ denote the set of all vectors in $\{0, \dots, r\}^T$ whose Hamming weight is at most $k$. For each $\bx = (x_1, \dots, x_T) \in \cX'$, let $D_{\bx}$ denote dataset where there are $n = 2\tau^*(k) + 1$ users and, for each $i \in [n]$, user $i$'s input $\bx_i = (x_{i, 1}, \dots, x_{i, T})$ is defined as follows:
\begin{align*}
x_{i, t} =
\begin{cases}
x_t &\text{ if } x_t \ne 0 \\
r + 1 + \lceil 3(i - 1)/n \rceil &\text{ otherwise.}
\end{cases}
\end{align*}
for all $t \in [T]$.
Notice here that $r + 1, r+ 2, r+3$ are never $\tau^*(k)$-heavy hitters. Thus, since the Hamming norm of $\bx$ is at most $k$, \Cref{assum:shifting-hh} is satisfied.

Now, let $O_{\bx}$ be the set of outcomes where $x_t$ is reported at time $t$ for all $t \in [T]$ such that $x_t \ne 0$ and no element in $[r] \setminus \{x_t\}$ is reported. 
From the utility guarantee of $\cM$, we have $\Pr[\cM(D_\bx) \in O_{\bx}] \geq 1/2$. From $\eps$-DP, we have $\Pr[\cM(\emptyset) \in O_\bx] \geq e^{-\eps \cdot n} / 2$. Meanwhile, since $O_{\bx}$ are disjoint for all $\bx \in \cX'$, we have
\begin{align*}
1 &\geq \sum_{\bx \in \cX'} \Pr[\cM(\emptyset) \in O_\bx] \\
&\geq |\cX'| \cdot e^{-\eps \cdot n} / 2 \\
&\geq r^{\lceil 0.1k \rceil} \binom{T}{\lceil 0.1k \rceil} \frac{e^{-\eps \cdot n}}{2} \\
&\geq \left(\frac{rT}{k}\right)^{0.01k} \frac{e^{-\eps \cdot n}}{2} \\
&> 1,
\end{align*}
where the last inequality follows from our choice of $n = 2\tau^*(k) + 1$. This completes our proof.
\end{proof}

\section{Counterexample for Pure-DP without Subsampling}
\label{app:counterexample}

Recall that, if we apply the (basic) composition theorem to the $L$ calls of exponential mechanisms in $\rem_{\eps, \cA}$ (\Cref{alg:repeated-em}), we get that the algorithm is $2L\eps$-DP. (Note that this is \emph{two-sided} DP.) The lemma below shows that we cannot hope to do much better than this bound. This lemma also gives a justification to our subsampling framework since we can achieve $\eps$-DP with subsampling (\Cref{thm:subsample-dp}). 

\begin{lemma}
There exists $\cA$ that satisfies \Cref{assum:monotone} and \Cref{assum:marginal-sensitivity} but that $\rem_{\eps, \cA}$ is not $\eps'$-DP for any $\eps' < L\eps$. 
\end{lemma}

\begin{proof}
In fact, we will use the instantiation for 1-decomposable submodular function as in \Cref{sec:smc}. Let $m = \left\lceil 1 
+ \frac{(e^\eps - 1)L}{e^{\eps} - e^{\eps'/L}}\right\rceil, D' = \emptyset$, and $D = \{x\}$, where $f_x$ is defined as $f_x(S) = \min\{1, |S \cap [m - L]|\}$ for $S \subseteq [m]$. Let $\cM$ denote the mechanism $\smcalg_{\eps, F_D}$ from \Cref{sec:smc} (which is an instantiation of $\rem_{\eps, \cA}$). Now, consider the output $o = (m - L + 1, \dots, m)$. Then, we have
\begin{align*}
\frac{\Pr[\cM(D') = o]}{\Pr[\cM(D) = o]}
&= \prod_{i \in [L]} \frac{\Pr[\EM_{\frac{\eps_0}{\Delta}}(\cC_i, q_i; D') = o_i]}{\Pr[\EM_{\frac{\eps_0}{\Delta}}(\cC_i, q_i; D) = o_i]} \\
&= \prod_{i \in [L]} \frac{\frac{1}{m-i+1}}{\frac{1}{e^\eps(m - L) + L - i + 1}} \\
&< \prod_{i \in [L]} e^{\eps'/L} \\
&= e^{\eps'},
\end{align*}
where the inequality is due to our choice of parameter $m$. Thus, the algorithm is not $\eps'$-DP.
\end{proof}

\section{Monotone Submodular Maximization over Matroid Constraint}
\label{app:submod-matroid}

As mentioned earlier, we will use the private version of the continuous greedy algorithm due to \citet{CNZ21}. 
To describe the algorithm, we need several additional definitions.

\begin{definition}[Multilinear Extension]
For a given set function $F: 2^{[m]} \to \R$, its multilinear extension $F^{\mult}: [0, 1]^m \to \R$ defined by
$$F^{\mult}(\by) = \sum_{S \subseteq [m]} F(S) \prod_{i \in S} y_i \prod_{i \in ([m] \setminus S)} (1 - y_i).$$
\end{definition}

For a given matroid $\cM = ([m], \cI)$, we write $\cP(\cM)$ to denote the \emph{matroid polytope},  which is the convex hull of all characteristics of all independent sets in $\cM$. The continuous greedy algorithm finds a vector in $\cP(\cM)$ with a large multilinear extension value. Since computing the multilinear extension directly is inefficient, it is only computed approximately. In~\cite{CNZ21}, this was done by randomly sampling $z^1, \dots, z^s$ uniformly and independently from $[0, 1]^m$, and then computing $$G^{\bz}(y) := \frac{1}{s} \sum_{j \in [s]} F(\{u \in [m] \mid z^j_u < y_u\}),$$ and using it as a proxy for the multilinear extension. The full algorithm is presented in \Cref{alg:contgreedy}.

\begin{algorithm}
\caption{$\contgreedy_{\eps_0, \eta, s}$~\cite{CNZ21}}
\label{alg:contgreedy}
\begin{algorithmic}
\PARAMETERS $\eps_0 > 0$, step size $\eta$, number of draws $s$
\REQUIRE Dataset $D \in \cX^*$
\STATE $T \gets \lceil 1/\eta \rceil$
\STATE $k \gets \rank(\cM)$
\STATE $z^1, \dots, z^s \sim [0, 1]^m$
\FOR{$t = 1, \dots, T$}
\STATE $B^{t, 0} \gets \emptyset$
\FOR{$i = 1, \dots, k$}
\STATE $\cC^{t, i} \gets \{u \in [m] \setminus B^{t, i - 1} \mid B^{t, i - 1} \cup \{u\} \in \cI\}$
\IF{$\cC^{t, i} = \emptyset$}
\STATE  $y^{t, i} \gets y^{t, i - 1}$
\ELSE
\STATE $q^{t, i}(u; D) \gets G^{\bz}_D(y^{t,i-1} + \eta \cdot \ind_u) - G^{\bz}_D(y^{t, i - 1})$ for all $u \in \cC^{t, i}$
\STATE $u^{t, i} \gets \EM_{\eps_0}(\cC^{t, i}, q^{t, i}; D)$
\STATE $y^{t, i} \gets y^{t, i - 1} + \eta \cdot \ind_{u^{t, i}}$
\STATE $B^{t, i} \gets B^{t, i - 1} \cup \{u^{t, i}\}$
\ENDIF
\ENDFOR
\STATE $y^{t+1, 0} \gets y^{t, k}$
\ENDFOR
\RETURN $y^{T, k}$
\end{algorithmic}
\end{algorithm}

Let $g^{\bz}_x := \frac{1}{s} \sum_{j \in [s]} f_x(\{u \in [m] \mid z^j_u < y_u\})$.
Notice that $G^{\bz}_D = \sum_{x \in D} g^{\bz}_x$. The utility guarantee of \Cref{alg:contgreedy} was shown in \cite{CNZ21}; we state this below.

\begin{theorem}[\citealt{CNZ21}] \label{thm:cont-greedy-util-cnn}
For any $\beta > 0$, let $s = 6 k^2 T^4 \cdot \log(m / \beta)$, then with probability $1 - \beta$, the output $y = y^{T, k}$ satisfies
\begin{align*}
F_D^{\mult}(y) \geq \left(1 - \frac{1}{e} - \eta\right) \max_{S \in \cI} \{ F_D(S) \} - O\left(\frac{k}{\eta \eps} \log\frac{m}{\eta \beta}\right).
\end{align*}
\end{theorem}

It is simple to see that $q^{t, i}$ is monotone.
Moreover, as observed in \cite{CNZ21}, the realized marginal sensitivity of the score $G_D^{\bz}$ is bounded:
\begin{observation}[\citealt{CNZ21}] \label{obs:marginal-sensitivity-contgreedy}
For any $z^1, \dots, z^s \in [0, 1]^m$, any neighboring datasets $D'$ and $D = D' \cup \{x\}$ and any selection of $(u^{t, i})_{t \in [T], i \in [k]}$, we have
\begin{align*}
&\sum_{t \in [T]} \sum_{i \in [k]} \left|q^{t, i}(u^{t, i}; D) - q^{t, i}(u^{t, i}; D')\right| \\
&= \sum_{t \in [T]} \sum_{i \in [k]} \left[\left(G^{\bz}_D(y^{t, i}) - G^{\bz}_{D}(y^{t, i - 1})\right) - \left(G^{\bz}_{D'}(y^{t, i}) - G^{\bz}_{D'}(y^{t, i - 1})\right)\right] \\
&= \sum_{t \in [T]} \sum_{i \in [k]} \left(g^{\bz}_x(y^{t, i}) - g^{\bz}_x(y^{t, i-1})\right) \\
&= g^{\bz}_x(y^{T, k}) - g^{\bz}_x(y^{1, 0}) \leq 1.
\end{align*}
\end{observation}

Combining these observations together with \Cref{thm:rem-add-dp} immediately yields the following:

\begin{lemma} \label{lem:fractional-submod}
For any $0 < \eps, \eta, \beta < 1$, there is an $\eps$-add-DP algorithm that for matroid submodular maximization that with probability $1 - \beta$ outputs $y$ such that
\begin{align*}
F_D^{\mult}(y) \geq \left(1 - \frac{1}{e} - \eta\right) \max_{S \in \cI} \{ F_D(S) \} - O\left(\frac{k}{\eta \eps} \log\frac{m}{\eta \beta}\right).
\end{align*}
\end{lemma}

\paragraph{Rounding.} To get from a fractional solution to an integral solution (i.e., a set), we recall the following rounding algorithm due to~\citet{CVZ10}. It should be noted that this rounding algorithm only depends on $y$ (and the matroid $\cM$) and does \emph{not} depend on the function $F$. 

\begin{definition}[\citealt{CVZ10}]
Let $\cM = ([m], \cI)$ be any matroid.
There exists a randomized rounding algorithm $\swapr$ that takes in $y \in \cP(\cM)$ and output a set $S \in \cI$ such that, for any submodular function $F$ and any $y \in [0, 1]^m$, we have
\begin{align*}
\E_{S \sim \swapr(\by)}[F(S)] \geq F^{\mult}(y).
\end{align*}
\end{definition}

We use a slightly different rounding procedure compared to \cite{CNZ21}: while they apply $\swapr$ once, we apply it multiple times and use the exponential mechanism to pick the best of them. This allows us to get high probability bound (compared to expected bound) on the approximation ratio and error. This is summarized below.

\begin{lemma} \label{lem:integral-submod}
For any $0 < \eps, \eta, \beta < 1$, there is an $\eps$-add-DP algorithm that for matroid submodular maximization that with probability $1 - \beta$ outputs $S^* \in \cI$ such that
\begin{align*}
F_D(S^*) \geq \left(1 - \frac{1}{e} - \eta\right) \max_{S \in \cI} \{ F_D(S) \} - O\left(\frac{k}{\eta \eps} \log\frac{m}{\eta \beta}\right).
\end{align*}
\end{lemma}

\begin{proof}
The algorithm works as follows:
\begin{itemize}[nosep]
\item Run the algorithm from \Cref{lem:fractional-submod} with parameters $\eps/2, \eta/2, \beta/3$ to get $y \in \cP(\cM)$.
\item Run $\swapr(y)$ $h = \lceil 10\log(3/\beta) / \eta \rceil$ times to arrive at sets $S^*_1, \dots, S^*_h$.
\item Use $(\eps/2)$-DP exponential mechanism based on the score $q(S^*_i; D) := F_D(S^*_i)$ to select $S^*$ from $S^*_1, \dots, S^*_h$ to (approximately) maximize $F_D(S^*)$.
\end{itemize}
Since the second step is just a post-processing of the result from the first step, we can apply composition across the two add-DP mechanisms in the first and last steps to conclude that this is $\eps$-add-DP.

As for the utility, recall from \Cref{lem:fractional-submod} that with probability $1 - \beta/3$, we have
\begin{align} \label{eq:fractional-submod-util}
F_D^{\mult}(y) \geq \left(1 - \frac{1}{e} - \frac{\eta}{2}\right) \max_{S \in \cI} \{ F_D(S) \} - O\left(\frac{k}{\eta \eps} \log\frac{m}{\eta \beta}\right).
\end{align}
For fixed $y$, let $\theta := F^{\mult}_D(y)$ and $\opt := \max_{S \in \cI} F_D(S)$. Since $\E_{S \sim \swapr(y)}[F_D(S)] = F_D^{\mult}(y) = \theta$ and $F_D(S) \leq \opt$ for all $S \in \cI$, Markov inequality implies that 
\begin{align*}
\Pr_{S \sim \swapr(y)}[F_D(S) \geq \theta - \eta \cdot \opt / 2] &= 1 - \Pr_{S \sim \swapr(y)}[\opt - F_D(S) > \opt + \eta \cdot \opt / 2 - \theta] \\
&\geq 1 - \frac{\opt - \theta}{\opt + \eta \cdot \opt / 2 - \theta} \\
&= \frac{\eta \cdot \opt / 2}{\opt + \eta \cdot \opt / 2 - \theta} \\
&\geq \eta / 4.
\end{align*}
Thus, by our choice of $h$, the following holds with probability at least $1 - \beta/3$:
\begin{align} \label{eq:repeated-candidates-util}
\max\{F_D(S^*_1), \dots, F_D(S^*_h)\} \geq F_D^{\mult}(y) - \eta \cdot \opt/2.
\end{align}
Finally, by the utility of the exponential mechanism (\Cref{thm:em}), with probability at least $1 - \beta / 3$, the following holds:
\begin{align} \label{eq:exp-mech-candidate-selection}
F_D(S^*) \geq \max\{F_D(S^*_1), \dots, F_D(S^*_h)\} - O\left(\frac{\log(h/\beta)}{\eps}\right).
\end{align}
When \eqref{eq:fractional-submod-util},\eqref{eq:repeated-candidates-util}, and \eqref{eq:exp-mech-candidate-selection} all hold (which happens with probability at least $1 - \beta$ due to the union bound), we have
\begin{align*}
F_D(S^*) \geq \left(1 - \frac{1}{e} - \eta\right) \max_{S \in \cI} F_D(S) - O\left(\frac{k}{\eta \eps} \log\frac{m}{\eta \beta}\right),
\end{align*}
which concludes our proof.
\end{proof}

We are now ready to prove the main theorem here (\Cref{thm:main-submod-matroid}) by appealing to concentration bounds (similar to the proof of \Cref{thm:smc}).

\begin{proof}[Proof of \Cref{thm:main-submod-matroid}]
Let $\cA_{\eps, \eta, \beta}$ denote the algorithm from \Cref{lem:integral-submod}. 
We simply run the subsampled version of this algorithm. More precisely, we use the algorithm $\cA_{\ln(2), \eta/2, \beta/2}^{\cS_p}$ where $p = 1 - e^{-\eps}$. The privacy guarantee immediately follows from \Cref{thm:rem-add-dp}.

To analyze the utility, let $D_{\subs} \sim \cS_p(D)$ denote the subsampled dataset that is fed as an input to $\cA_{\ln(2), \eta/2, \beta/2}$ and let $S^*$ denote the output set. By the utility guarantee of $\cA$ (\Cref{lem:integral-submod}), with probability $1 - \beta/2$, we have
\begin{align} \label{eq:integral-submod}
F_{D_{\subs}}(S^*) \geq \left(1 - \frac{1}{e} - \frac{\eta}{2}\right) \cdot \max_{S \subseteq I} F_{D_{\subs}}(S) - O\left(\frac{k}{\eta} \log\frac{m}{\eta \beta}\right).
\end{align}
Furthermore, applying the Chernoff bound (\Cref{thm:chernoff}) with $Z_x := f_x(S) \cdot \ind[x \in D_{\subs}], \mu = p \cdot F_D(S), \alpha = 0.1\eta, \zeta = \frac{100 k \log(m/\beta)}{\eta}$ together with a union bound over all sets $S \in \binom{\cU}{\leq k}$, we can conclude that the following holds simultaneously for all $S \in \binom{\cU}{\leq k}$ with probability at least $1 - \beta/2$:
\begin{align}
F_{D_{\subs}}(S) &\geq (1 - \alpha)p \cdot F_D(S) - \zeta, \label{eq:submod-lb} \\
F_{D_{\subs}}(S) &\leq (1 + \alpha)p \cdot F_D(S) + \zeta. \label{eq:submod-ub}
\end{align}
When \eqref{eq:integral-submod}, \eqref{eq:submod-lb}, and \eqref{eq:submod-ub} all hold, we have
\begin{align*}
F_D(S^*)  
&\overset{\eqref{eq:submod-ub}}{\geq} \frac{1}{(1 + \alpha)p} F_{D_{\subs}}(S^*) - \zeta / p \\
&\overset{\eqref{eq:integral-submod}}{\geq} \frac{1}{(1 + \alpha)p} \cdot \left(1 - \frac{1}{e} - \frac{\eta}{2}\right) \cdot \max_{S \in \cI} F_{D_{\subs}}(S) - O\left(\frac{k \log \left(\frac{m}{\eta\beta}\right)}{\eta p}\right) - \zeta / p \\
&\overset{\eqref{eq:submod-lb}}{\geq} \frac{1 - \alpha}{1 + \alpha} \cdot \left(1 - \frac{1}{e} - \frac{\eta}{2}\right) \cdot \max_{S \in \cI} F_{D}(S) - O\left(\frac{k \log \left(\frac{m}{\eta\beta}\right)}{\eta p}\right) - 2\zeta / p  \\
&\geq \left(1 - \frac{1}{e} - \eta\right) \cdot \max_{S \in \cI} F_{D}(S) - O\left(\frac{k \log \left(\frac{m}{\eta\beta}\right)}{\eta \eps}\right), 
\end{align*}
where the last inequality follows from our choice of parameters.
\end{proof}

\section{Set Cover: Proof of \Cref{thm:set-cover-main}}
\label{app:set-cov}

\begin{proof}[Proof of \Cref{thm:set-cover-main}]
We use $\setcovalg_{\eps_0}$. To see the privacy guarantee of the algorithm, note that the computation of $\tn$ is $0.5\eps$-DP. Meanwhile, the $r$th round is (a post-processing of) an instantiation of $\rat_{\ln(2), \cA}^{S_{p_r}}$ with $h_{r, i}(D) = \left|(S_i \cap D) \setminus \left(\bigcup_{j' < j} S_{\pi(j')} \right)\right|$. It is simple to verify that this satisfies Assumptions~\ref{assum:monotone-at} and~\ref{assum:marginal-sensitivity-at}. Thus, \Cref{thm:subsample-dp-rat} implies that the $r$th round is $\eps_r$-DP. Applying (basic) composition theorem, we get that the entire algorithm is $\eps'$-DP for
\begin{align*}
\eps' = 0.5\eps + \sum_{r \in \N} \eps_r \leq \eps, 
\end{align*}
as desired.

Next, we analyze the approximation guarantee. If $n \leq  \frac{1000\log(m/\beta)}{\eps}$, then any output will yield an approximation ratio at most $n \leq O\left(\frac{\log(m/\beta)}{\eps}\right)$. Hence, we may assume w.l.o.g. that $n > \frac{1000\log(m/\beta)}{\eps}$. In this case, standard tail bounds for Laplace noise shows that with probability $1 - \beta/3$, we have
\begin{align} \label{eq:n-est}
0.5\tn \leq n \leq 2\tn. 
\end{align}
We will condition on this event happening for the remainder of the analysis.

We write $S_{\pi(<j)}$ as a shorthand for $\bigcup_{j' < j} S_{\pi(j')}$. For all $r = \{0, \dots, R\}$, let $\beta_r = \beta / 2^{R-r}$ and let $j_r$ denote the value of $j$ at the end of round $r$. Furthermore, let $q = 100 \cdot \setcov(\cC, \cS)$. 

Let us fix $r \in [R]$. By tail bounds for exponential noise, the following holds for all $i \in \cI$ with probability $1 - \beta_r/6$:
\begin{align} \label{eq:noise-bound-sc}
\theta_{r, i} \leq \log(6m/\beta_r).
\end{align}
Furthermore, applying the Chernoff bound (\Cref{thm:chernoff}) with $Z_x := \ind\left[x \in (D \setminus S_{T \cup \pi(<j_{r-1})})\right] \cdot \ind[x \in D_{\subs, r}], \mu = p \cdot |D \setminus S_{T \cup \pi(<j_{r-1})}|, \alpha = 0.3, \zeta_r = 10 \cdot \ln(12 m^q/\beta_r)$ together with a union bound over all sets $T \in \binom{[m]}{\leq q}$, we can conclude that the following hold simultaneously for all $T \in \binom{[m]}{\leq q}$ with probability at least $1 - \beta_r/6$:
\begin{align} 
|D_{\subs, r} \setminus S_{T \cup \pi(<j_{r-1})}| &\geq 0.7p_r |D \setminus S_{T \cup \pi(<j_{r-1})}| - \zeta_r, \label{eq:sc-lb-opt} \\
|D_{\subs, r} \setminus S_{T \cup \pi(<j_{r-1})}| &\leq 1.3p_r |D \setminus S_{T \cup \pi(<j_{r-1})}| + \zeta_r. \label{eq:sc-ub}
\end{align}

Note that, by a union bound, the probability that \eqref{eq:n-est} holds and \eqref{eq:noise-bound-sc}, \eqref{eq:sc-lb-opt}, and \eqref{eq:sc-ub} hold for all $r \in [R]$ is at least
\begin{align*}
1 - \frac{\beta}{3} - \sum_{r \in [R]} \frac{2\beta_r}{6} \geq 1 - \beta.
\end{align*}

For the remainder of the analysis, we will assume that \eqref{eq:n-est} holds and \eqref{eq:noise-bound-sc}, \eqref{eq:sc-lb-opt}, and \eqref{eq:sc-ub} hold for all $r \in [R]$. A crucial claim used to bound the approximation ratio is stated below. 
\begin{claim} \label{clm:inductive-sc}
For all $r \in [R]$, we have
\begin{align} \label{eq:selected-set-per-bound}
j_r - j_{r - 1} \leq q,
\end{align}
and
\begin{align} \label{eq:num-remaining-elements}
|D \setminus S_{\pi(<j_r)}| \leq \frac{0.1 q \cdot \tau_r}{p_r}.
\end{align}
\end{claim}

Before we prove \Cref{clm:inductive-sc}, let us show how to use this to finish the proof of the approximation guarantee. The number of total sets chosen is at most
\begin{align*}
j_R + |D \setminus S_{\pi(< j_R)}| &\leq R \cdot q + \frac{0.1 q \cdot \tau_R}{p_R} \\
&\leq O\left(\left(\log n + \frac{\log m}{\eps}\right)\right) \cdot \optsc(\cC, \cS),
\end{align*}
where the first inequality is due to \Cref{clm:inductive-sc} and the second inequality is due to our choice of parameters. Thus, we can conclude that the approximation ratio is $O\left(\log n + \frac{\log m}{\eps}\right)$ as desired.
\end{proof}

We now prove \Cref{clm:inductive-sc}.

\begin{proof}[Proof of \Cref{clm:inductive-sc}]
Before we proceed, let us state a few  inequalities that will be subsequently useful. First, note that $\tau_r = \frac{\tn}{2^r} \geq n_{\min} \cdot 2^{R - r}$. From this, we can derive
\begin{align} \label{eq:noise-bound-sc-concrete}
\theta_{r,i} \overset{\eqref{eq:noise-bound-sc}}{\leq} \log(6m/\beta_r) = \log(m/\beta) + (R - r + 3) \leq 0.005 \tau_r,
\end{align}
and
\begin{align} \label{eq:concen-add-error-concrete}
\zeta_r = 10 \cdot \ln(12 m^q/\beta_r) \leq 10 q \ln(m / \beta_r) \leq 0.05 q \tau_r.
\end{align}

For brevity, we call the two statements in the claim $P(r)$. We will prove $P(r)$ by induction. For convenience, we also define $P(0)$ where we let $j_{-1} = j_0, \tau_0 = \tn, \eps_0 = \frac{\eps}{4 \cdot 2^R}$ and $p_0 = 1 - e^{-\eps_0}$.

\paragraph{Base case.} For $r = 0$, \eqref{eq:selected-set-per-bound} trivially holds. Meanwhile, \eqref{eq:num-remaining-elements} follows immediately from \eqref{eq:n-est}.

\paragraph{Inductive Step.} Suppose that $P(r - 1)$ holds for some $r \in \N$. To see that \eqref{eq:selected-set-per-bound} holds, note that \eqref{eq:num-remaining-elements} from $P(r - 1)$ and \eqref{eq:sc-ub} with $T = \emptyset$ implies that
\begin{align*}
|D_{\subs, r} \setminus S_{\pi(<j_{r - 1})}| &\overset{\eqref{eq:sc-ub}}{\leq} 1.3 p_r |D \setminus S_{\pi(<j_{r - 1})}| + \zeta_r \\
&\overset{\eqref{eq:num-remaining-elements}}{\leq} 1.3 p_r \cdot \frac{0.1 q \tau_{r - 1}}{p_{r - 1}} + \zeta_r.
\end{align*}
Note that $\frac{p_r}{p_{r - 1}} = \frac{1 - e^{-\eps_r}}{1 - e^{-\eps_r/2}} = 1 + e^{-\eps_{r}/2} \leq 2$ and that $\tau_{r - 1} = 2 \tau_r$. Plugging these together with \eqref{eq:concen-add-error-concrete} into the above, we get
\begin{align*}
|D_{\subs, r} \setminus S_{\pi(<j_{r - 1})}| \leq 0.52 q \tau_r + 0.05 q \tau_r \leq 0.6 \tau_r.
\end{align*}

From \eqref{eq:noise-bound-sc-concrete}, every set chosen in round $r$ covers at least $0.995\tau_r$ additional uncovered elements in $D_{\subs, r}$. As a result, the number of sets chosen in round $r$ (which is equal to $j_r - j_{r - 1}$) is at most
\begin{align*}
\frac{|D_{\subs, r} \setminus S_{\pi(<j_{r - 1})}|}{0.995\tau_r} \leq \frac{0.6 q \tau_r}{0.995\tau_r} < q,
\end{align*}
proving \eqref{eq:selected-set-per-bound}.

Next, we will prove \eqref{eq:num-remaining-elements}. First, observe that at the end of the rounds, since $\theta_{r, i}$ are all non-negatives, the remaining sets $i \notin \cI$ must satisfy $|D_{\subs, r} \cap S_i| < \tau_r$. This implies that the number of remaining elements $|D_{\subs, r} \cap S_{\pi(<j_r)}|$ is at most $\tau_r \cdot \optsc(\cC, \cS) = 0.01 \tau_r \cdot q$. From \eqref{eq:selected-set-per-bound}, we have $j_r - j_{r - 1} \leq q$. Thus, we may apply \eqref{eq:sc-lb-opt} with $T = \{\pi(j_{r - 1}), \dots, \pi(j_r - 1)\}$ to arrive at
\begin{align*}
|D \setminus S_{\pi(<j_r)}| < \frac{|D_{\subs, r} \setminus S_{\pi(<j_r)}| + \zeta_r}{0.7p_r}
\leq \frac{0.01 \tau_r \cdot q + 0.05 q\tau_r}{0.7p_r} < \frac{0.1 q \cdot \tau_r}{p_r},
\end{align*}
proving \eqref{eq:num-remaining-elements}.

Thus, $P(r)$ holds for all $r \in [R]$. This completes our proof.
\end{proof}

\section{On Set Cover via Subsampled Repeated EM}
\label{app:set-cov-non-opt}

The algorithm of~\citet{GLMRT10} for Set Cover is exactly the same as their algorithm for submodular maximization ($\smcalg_{\eps_0}$) except with $k = m$ and $F^{\setcov}_D(T) := |\bigcup_{i \in T} S_i|$.
Finally, the output permutation is just $\pi = (c^*_1, \dots, c^*_m)$. Similar to \Cref{sec:smc}, it is possible to use the subsampled version of this algorithm for Set Cover. Unfortunately, here we can only show an approximation ratio of $O\left(\frac{\log n \log m}{\eps}\right)$ instead of the optimal $O\left(\log n + \frac{\log m}{\eps}\right)$ that we presented in \Cref{sec:set-cov}:

\begin{theorem} \label{thm:set-cov-non-opt}
For any $0 < \eps, \beta, \eta \leq 1$, $\smcalg_{\ln(2), F_D^{\setcov}}^{\cS_p}$ where $p = 1 - e^{-\eps}$ is a polynomial-time $\eps$-DP algorithm for Set Cover that achieves $O\left(\frac{\log n \log(m/\beta)}{\eps}\right)$-approximation with probability $1 - \beta$.
\end{theorem}

We state the utility guarantee of $\smcalg$ from \cite{GLMRT10} in a fine-grained fashion below.

\begin{theorem}[\citealt{GLMRT10}] \label{thm:sc-gupta-util}
With probability $1 - \beta$, the output $\pi$ from $\smcalg_{\eps_0, F_D^{\setcov}}$ satisfies the following: there exists $r = O(\optsc(\cU, \cS) \cdot \log n)$ such that $\left|S_i \setminus \left(\bigcup_{j \in [r]} S_{\pi(j)}|\right)\right| \leq O\left(\frac{\log(m/\beta)}{\eps_0}\right)$ for all $i \in [m]$.
\end{theorem}

\begin{proof}[Proof of \Cref{thm:set-cov-non-opt}]
The privacy guarantee follows in a similar manner as in the proof of \Cref{thm:smc}.

To analyze the utility, let $D_{\subs} \sim \cS_p(D)$ denote the subsampled dataset that is fed as an input to $\smcalg_{\ln(2), F_D^{\setcov}}$; note that we view $D_{\subs}$ as the universe $\cU$ and let $S'_i := D_{\subs} \cap S_i$ for all $i \in [m]$. Furthermore, let us abbreviate $\left(\bigcup_{j \in T} S_j\right)$ as $S_{T}$,  $\left(\bigcup_{j \in T} S'_j\right)$ as $S'_{T}$; furthermore, we abbreviate $S_{\{\pi(1), \dots, \pi(r)\}}$ as $S_{\pi(\leq r)}$ and similarly define $S'_{\pi(\leq r)}$.
By the utility guarantee of $\smcalg$ (\Cref{thm:sc-gupta-util}), with probability $1 - \beta/2$, there exists $r = O(\optsc(D_{\subs}, \cS') \cdot \log n) \leq O(\optsc(D, \cS) \cdot \log n)$ where
\begin{align} 
|S'_i \setminus S'_{\pi(\leq r)}| \leq O(\log(m/\beta)) &&\forall i \in [m].
\end{align}
This means that each $S'_i$ can cover at most $O(\log(m/\beta))$ elements from $D_{\subs} \setminus S'_{\pi(\leq r)}$, which implies
\begin{align} \label{eq:sc-util-subsampled}
\Omega\left(\frac{|D_{\subs} \setminus S'_{\pi(\leq r)}|}{\log(m/\beta)}\right) \leq \optsc(D_{\subs}, \cS') \leq \optsc(D, \cS).
\end{align}

Furthermore, applying the Chernoff bound (\Cref{thm:chernoff}) with $Z_x := \ind\left[x \in (D \setminus S_{T})\right] \cdot \ind[x \in D_{\subs}], \mu = p \cdot |D \setminus S_{T}|, \alpha = 0.1, \zeta = 2000 r \log(m/\beta)$ together with a union bound over all sets $T \in \binom{\cU}{\leq r}$, we can conclude that the following holds simultaneously for all $T \in \binom{\cU}{\leq r}$ with probability at least $1 - \beta/2$:
\begin{align} 
|D_{\subs} \setminus S'_T| &\geq (1 - \alpha)p |D \setminus S_T| - \zeta. \label{eq:sc-lb}
\end{align}
When \eqref{eq:sc-util-subsampled} and \eqref{eq:sc-lb} both hold, we have
\begin{align*}
|D \setminus S_{\pi(\leq r)}| &\overset{\eqref{eq:sc-lb}}{\leq} \frac{1}{(1 - \alpha)p} \left(\zeta + |D_{\subs} \setminus S'_{\pi(\leq r)}|\right) \\
&\overset{\eqref{eq:sc-util-subsampled}}{\leq} \frac{1}{(1 - \alpha)p} \left(\zeta + \log(m/\beta) \cdot \optsc(D, \cS)\right) \\
&\leq O\left(\frac{\log n \log(m/\beta)}{\eps} \cdot \optsc(D, \cS)\right),
\end{align*}
where the last inequality follows from our choice of $r, p, \alpha$.

Finally, observe that the number of sets that are chosen after $r$ is at most $|D \setminus S_{\pi(\leq r)}|$. Thus, in total the number of sets chosen is at most $r + |D \setminus S_{\pi(\leq r)}| \leq O\left(\frac{\log n \log(m/\beta)}{\eps} \cdot \optsc(D, \cS)\right)$.
\end{proof}

\section{\boldmath Metric $k$-Means and $k$-Median}
\label{app:clustering}

We write $\opt^q_k(D)$ to denote $\min_{S \subseteq [m] \atop |S| = k} \cost^q(S; D)$. Furthermore, for notational convenience, we let $\cost^q(\emptyset; D) = n$ (i.e., we think of $\min_{c \in S} d(c, x)^q$ as being 1 when $S$ is empty).

\subsection{One-Sided DP Algorithms}

\subsubsection{Bicriteria Approximation}

It will be more convenient to start from a one-sided DP algorithm. In fact, we will first give a bicriteria approximation algorithm where the output set size is $O(k \log n)$, as stated below. Note that the algorithm is fairly similar to that of \citet{JNN21}. However, ours is simpler since we can apply the repeated exponential mechanism (\Cref{alg:repeated-em}) directly, while their algorithm uses maximum $k$-coverage algorithm as a black-box (and thus they have to handle different distance scales explicitly).

\begin{theorem} \label{thm:clustering-bicriteria-one-sided}
For any $0 < \eps, \beta < 1$, there is an $\eps$-add-DP algorithm that, with probability $1 - \beta$, outputs a set $T 
\subseteq [m]$ of size at most $O(k \log n)$ such that $\cost^q(T; D) \leq \opt^q_k(D) + O\left(\frac{k \cdot \log(m/\beta)}{\eps}\right)$.
\end{theorem}

\begin{proof}
If $2k \ln n \geq m$, then the bound is trivial by outputting $T = [m]$. Otherwise, we use the $\rem_{\eps, \cA}$ algorithm with $L = \lceil 2k \ln n \rceil$ and the candidate sets and scoring functions selected as follows:
\begin{itemize}[nosep]
\item $\cC_i$ is the set of remaining elements $[m] \setminus \{c^*_1, \dots, c^*_{i - 1}\}$ 
\item $q_i(c; D)$ is the reduction in the cost after adding $c$, i.e., $\cost^q(\{c^*_1, \dots, c^*_{i - 1}\}; D) - \cost^q(\{c^*_1, \dots, c^*_{i - 1}, c\}; D)$. 
\end{itemize}
It is simple to see that \Cref{assum:monotone} and \Cref{assum:marginal-sensitivity} are satisfied and thus the algorithm is $\eps$-add-DP by \Cref{thm:rem-add-dp}.

To analyze its utility, let $T_i = \{c^*_1, \dots, c^*_i\}$ denote the solution set maintained by the algorithm at time step $i$. Invoking \Cref{thm:em} together with a union bound, we can conclude that, with probability $1 - \beta$, the following holds for all $i \in [L]$.
\begin{align} \label{eq:em-clustering-cost-reduction}
\left(\cost^q(T_i; D) - \cost^q(T_{i + 1}; D)\right) \geq \max_{c \in \cC_i} \left(\cost^q(T_i; D) - \cost^q(T_i \cup \{c\}; D)\right) - \kappa \cdot \frac{\log(m/\beta)}{\eps},
\end{align}
where $\kappa > 0$ is a constant. For the remainder of the proof, we will assume that this event holds.

Let $\Psi_i = \cost^q(T_i; d) - \opt^q_k(D)$ denote the difference between the cost of the solution at the $i$th step and the optimal solution (among $k$-size subsets). Let $S^*$ denote the optimal solution, i.e., $|S^*| = k$ such that $\cost^q(S^*; D) = \opt^q_k(D)$. If $\Psi_i \geq 0$, then it is simple to see that $\sum_{c \in S^*} \left(\cost^q(T_i; D) - \cost^q(T_i \cup \{c\}; D)\right) \geq \Psi_i$. This implies that $\max_{c \in \cC_i} \left(\cost^q(T_i; D)-\cost^q(T_{i} \cup \{c\}; D)\right) \geq \frac{\Psi_i}{k}$. Thus, if $\Psi_i \geq k \cdot \left(2\kappa \cdot \frac{\log(m/\beta)}{\eps}\right)$, then we can apply \eqref{eq:em-clustering-cost-reduction} to conclude that $\left(\cost^q(T_{i}; D) - \cost^q(T_{i+1}; D)\right) \geq \frac{\Psi_i}{2k}$. Rearranging, this gives
\begin{align} \label{eq:gap-exp-dec-clustering}
\Psi_{i+1} \leq \left(1 - \frac{1}{2k}\right) \Psi_i,
\end{align}

Thus, either we have $\Psi_i < k \cdot \left(2\kappa \cdot \frac{\log(m/\beta)}{\eps}\right)$ for some $i \in [L]$ which immediately satisfies the desired accuracy bound, or we can repeatedly apply \eqref{eq:gap-exp-dec-clustering} to arrive at
\begin{align*}
\Psi_L \leq \left(1 - \frac{1}{2k}\right)^{L} \Psi_0 \leq e^{-0.5L/k} \cdot n \leq 1,
\end{align*}
where the last inequality is due to our choice of $L$. As such, the accuracy guarantee also holds in this case.
\end{proof}

\subsubsection{From Bicriteria to True Approximation}

We can then go from bicriteria approximation to true approximation (i.e., output set size $k$) via standard techniques.

\begin{theorem}  \label{thm:clustering-approx-one-sided}
For any $0 < \eps, \beta < 1$, there is an $\eps$-add-DP algorithm for metric $k$-median and metric $k$-means that achieves $O(1)$-approximation and $O\left(\frac{k \log(mn/\beta)}{\eps}\right)$-error w.p. $1 - \beta$.
\end{theorem}

\begin{proof}
The algorithm works as follows.
\begin{itemize}[nosep]
\item Run the $\eps/2$-add-DP algorithm from \Cref{thm:clustering-bicriteria-one-sided} to obtain $T \subseteq [m]$ of size $O(k \log n)$.
\item Create a histogram $(\tih_t)_{t \in T}$ as follows. First, let $\hath_t = |\{x \in D \mid t = \argmin_{t' \in T} d(t', x)\}|$ (where ties are broken arbitrarily in $\argmin$). Then, sample $\theta_t \sim \Exp(\eps/2)$ independently and let $\tih_t = \hath_t + \theta_t$.
\item Let $\tD$ be the dataset that, for each $t \in T$, contains $\tih_t$ copies of $t$. Run any (non-private) $O(1)$-approximation algorithm (e.g., \cite{AryaGKMMP04} for $k$-median and \cite{KanungoMNPSW04} for $k$-means) on $\tD$ to produce a solution $S$. Then, output $S$.
\end{itemize}
Since the second step is an application of the exponential-noise mechanism, it satisfies $(\eps/2)$-add-DP. As a result, by applying the basic composition theorem and viewing the last step as a post-processing, we can conclude that the entire algorithm is $\eps$-add-DP.

By \Cref{thm:clustering-bicriteria-one-sided}, with probability $1 - \beta/2$, we have
\begin{align} \label{eq:bicri}
\cost^q(T; D) \leq \opt^q_k(D) + O\left(\frac{k \cdot \log(m/\beta)}{\eps}\right)
\end{align}
Furthermore, by standard concentration of sum of exponential random variables (see e.g.,~\cite{brown2011wasted}), the following holds with probability $1 - \beta/2$:
\begin{align} \label{eq:sum-exp-noise}
\sum_{t \in T} \theta_t \leq O\left(\frac{k \log(n/\beta)}{\eps}\right).
\end{align}
Henceforth, we will assume that \eqref{eq:bicri} and \eqref{eq:sum-exp-noise} hold. 

Let $S^*$ denote the optimal solution in the original dataset $D$, i.e., $|S^*| = k$ such that $\cost^q(S^*; D) = \opt^q_k(D)$. Furthermore, let $\hD$ be the dataset that, for each $t \in T$, contains $\hath_t$ copies of $t$. By the guarantee of the non-private approximation algorithm, we have 
\begin{align}
\cost^q(S; \tD) &\leq O(1) \cdot \opt^q(\tD) \nonumber \\
&\leq O(1) \cdot \cost^q(S^*; \tD) \nonumber \\
&\overset{\eqref{eq:sum-exp-noise}}{\leq} O(1) \cdot \cost^q(S^*; \hD) + O\left(\frac{k \log(n/\beta)}{\eps}\right) \nonumber \\
&\overset{(\clubsuit)}{\leq} O(1) \cdot \left(\cost^q(S^*; D) + \cost^q(T; D)\right) + O\left(\frac{k \log(n/\beta)}{\eps}\right) \nonumber \\
&\overset{\eqref{eq:bicri}}{\leq} O(1) \cdot \opt^q_k(D) + O\left(\frac{k \log(mn/\beta)}{\eps}\right), \label{eq:cost-on-synth}
\end{align}
where $(\clubsuit)$ follows from the $q$th power triangle inequality, i.e., $d(u, v)^q \leq 2^q (d(u, w)^q + d(w, v)^q)$ for all $u, v, w \in [m]$.

Finally, we have
\begin{align*}
\cost^q(S; D)
&\overset{(\spadesuit)}{\leq} O(1) \cdot \left(\cost^q(S; \hD) + \cost^q(T; D)\right) \\
&\overset{\eqref{eq:sum-exp-noise},\eqref{eq:bicri}}{\leq}  O(1) \cdot \left(\cost^q(S; \tD) + \frac{k \log(mn/\beta)}{\eps}\right) \\
&\overset{\eqref{eq:cost-on-synth}}{\leq} O(1) \cdot \left(\opt^q_k(D) + \frac{k \log(mn/\beta)}{\eps}\right),
\end{align*}
where $(\spadesuit)$ again follows from the $q$th power triangle inequality.
\end{proof}

\subsection{From One-Sided to Two-Sided DP}

Finally, we go from one-sided to two-sided DP using \Cref{thm:subsample-dp} similar to the previous analyses.

\begin{proof}[Proof of \Cref{thm:clustering}]
Let $\cA$ denote a $\ln(2)$-add-DP algorithm from \Cref{thm:clustering-approx-one-sided}. We use the algorithm subsampled version of this algorithm, i.e., $\cA^{\cS_p}$ for $p = 1 - e^{-\eps}$. \Cref{thm:subsample-dp} immediately implies that this algorithm is $\eps$-DP as desired.

For the utility, let $D_{\subs} \sim \cS_p(D)$ denote the subsampled dataset that is fed as an input to $\cA$ and let $T^*$ denote the output set. 
From \Cref{thm:clustering-approx-one-sided}, w.p. $1 - \beta/2$, we have
\begin{align} 
\cost^q(T^*; D_{\subs}) &\leq O(1) \cdot \opt^q_k(D_{\subs})  + O\left(k \log \left(\frac{mn}{\beta}\right)\right). \label{eq:clustering-util-subsampled}
\end{align}
Furthermore, applying the Chernoff bound (\Cref{thm:chernoff}) with $Z_x := \left(\min_{c \in T} d(c, x)^q \right) \cdot \ind[x \in D_{\subs}], \mu = p \cdot \cost^q(T; D), \alpha = 0.1, \zeta = 20000 k \log(m/\beta)$ together with a union bound over all sets $T \in \binom{[m]}{\leq k}$, we can conclude that the following hold simultaneously for all $T \in \binom{[m]}{\leq k}$ with probability at least $1 - \beta/2$:
\begin{align}
\cost^q(T; D_{\subs}) &\geq (1 - \alpha)p \cdot \cost^q(T; D) - \zeta \label{eq:clustering-lb}, \\
\cost^q(T; D_{\subs}) &\leq (1 + \alpha)p \cdot \cost^q(T; D) + \zeta \label{eq:clustering-ub}.
\end{align}
When \eqref{eq:clustering-util-subsampled}, \eqref{eq:clustering-lb}, and \eqref{eq:clustering-ub} all hold, we have
\begin{align*}
\cost^q(T^*; D) 
&\overset{\eqref{eq:clustering-lb}}{\leq} \frac{1}{(1 - \alpha)p} \cost^q(T^*; D_{\subs}) + \frac{\zeta}{(1 - \alpha)p} \\
&\overset{\eqref{eq:clustering-util-subsampled}}{\leq} O\left(\frac{1}{(1 - \alpha)p}\right) \cdot \opt^q_k(D_{\subs})  + O\left(\frac{k \log \left(\frac{mn}{\beta}\right)}{p(1 - \alpha)}\right) + \frac{\zeta}{(1 - \alpha)p} \\
&\overset{\eqref{eq:clustering-ub}}{\leq} O\left(\frac{1 + \alpha}{1 - \alpha}\right) \cdot \opt^q_k(D)  + O\left(\frac{k \log \left(\frac{mn}{\beta}\right)}{p(1 - \alpha)}\right) + O\left(\frac{\zeta}{(1 - \alpha)p}\right) \\
&\leq O\left(1\right) \cdot \opt^q_k(D)  + O\left(\frac{k \log \left(\frac{mn}{\beta}\right)}{\eps}\right),
\end{align*}
which concludes our proof.
\end{proof}

\end{document}